\pdfoutput=1
\documentclass[a4paper,UKenglish,cleveref,
numberwithinsect, pdfa]{lipics-v2021}
\hideLIPIcs%
\title{Two Variables and Two-Dimensional Structures}
\titlerunning{Two Variables and Two-Dimensional Structures (Under submission)}

\author{Bartosz Jan Bednarczyk}
{Institute of Logic and Computation, TU Wien, Austria
\and Computer Science Department, University of Wrocław, Poland
\and \url{https://bartoszjanbednarczyk.github.io}}{bartek@cs.uni.wroc.pl}
{https://orcid.org/0000-0002-8267-7554}
{}

\author{Daumantas Kojelis}
{Leipzig University, Germany
\and \url{https://daumantaskojelis.github.io/}}{daumantas.kojelis@uni-leipzig.de}
{https://orcid.org/0000-0002-1632-9498}
{}

\author{Krzysztof Szymański}
{Computer Science Department, University of Wrocław, Poland}{338068@uwr.edu.pl}
{}
{}

\authorrunning{B. J. Bednarczyk and D. Kojelis and K. Szymański}
\Copyright{Bartosz Jan Bednarczyk and Daumantas Kojelis and Krzysztof Szymański}



\funding{B.~Bednarczyk was supported in full by the Austrian Science Fund (FWF) project FINGO 10.55776/ESP2407425, and D.~Kojelis by DFG project LU 1417/4-1. Results from Sections~\ref{sec:one-order-one-succ}--\ref{sec:orderinv} were presented in the bachelor's thesis of K.~Szymański, written under the supervision of B.~Bednarczyk at the University of Wrocław.}

\usepackage{lineno}

\usepackage{ellipsis}
\usepackage[abbreviations,british]{foreign}
\usepackage{hyperref}
\usepackage[usenames,dvipsnames]{xcolor}
\usepackage{float}
\usepackage{booktabs}

\definecolor{darkmidnightblue}{rgb}{0.0, 0.2, 0.4}
\definecolor{persianplum}{rgb}{0.44, 0.11, 0.11}

\hypersetup{
  colorlinks  = true,
  citecolor   = persianplum,
  urlcolor    = persianplum,
  linkcolor   = persianplum
}

\usepackage{tikz}
\usepackage{todonotes}
\usepackage{xspace}

\usepackage{mathtools}
\usepackage{amsmath}
\usepackage{amsthm}
\usepackage{stmaryrd}
\theoremstyle{plain}
\newtheorem*{theorem*}{Theorem}
\newtheorem*{lemma*}{Lemma}
\newtheorem*{observation*}{Observation}
\newtheorem*{proposition*}{Proposition}
\newtheorem*{fact*}{Fact}
\newtheorem{fact}[theorem]{Fact}

\makeatletter
\def\desclabel#1#2{\begingroup
\def\@currentlabel{#1}%
\protected@edef\cref@currentlabel{[desclabel][][]#1}%
#1\label{#2}\endgroup
}
\makeatother
\crefname{desclabel}{condition}{conditions}
\Crefname{desclabel}{Condition}{Conditions}

\newenvironment{namedlist}[1][(Cycle)]
  {\begin{list}{}{%
     \setlength{\labelsep}{0.5em}%
     \settowidth{\labelwidth}{#1}%
     \setlength{\leftmargin}{\labelwidth}%
     \addtolength{\leftmargin}{\labelsep}%
     \setlength{\itemsep}{0pt}%
     \setlength{\parsep}{0pt}%
     \setlength{\topsep}{4pt plus 1pt minus 1pt}%
     }}
  {\end{list}}

\DeclareMathAlphabet{\mathdutchcal}{U}{dutchcal}{m}{n}
\SetMathAlphabet{\mathdutchcal}{bold}{U}{dutchcal}{b}{n}
\DeclareMathAlphabet{\mathdutchbcal}{U}{dutchcal}{b}{n}

\definecolor{ao(english)}{rgb}{0.0, 0.5, 0.0}

\usepackage{fontawesome5}
\usepackage{xurl}

\usetikzlibrary{arrows, arrows.meta, decorations.markings, shapes, calc}
\colorlet{jaune}{yellow!80!green}

\colorlet{vert}{green!45!black}
\def\envert#1{\textcolor{vert}{#1}}
\colorlet{bleu}{blue!70!black}

\colorlet{rouge}{red!80!black}

\tikzstyle{minirond}=[draw, circle, minimum height=2mm, inner sep=0pt]
\tikzstyle{ptrond}=[draw, circle, minimum height=2.5mm]
\tikzstyle{medrond}=[draw, circle, minimum height=5mm]
\tikzstyle{rond}=[draw, circle, minimum height=7mm]

\tikzstyle{carre}=[draw,minimum width=6mm,minimum height=6mm]
\tikzstyle{medcarre}=[draw,minimum width=4mm,minimum height=4mm]
\tikzstyle{ptcarre}=[draw,minimum width=2.8mm,minimum height=2.8mm]
\tikzstyle{minicarre}=[draw,minimum width=1.5mm,minimum height=1.5mm,inner sep=0pt]
\tikzstyle{rouge}=[draw=red,fill=red!20!white]
\tikzstyle{vert}=[draw=green!80!black,fill=green!80!black!20!white]
\tikzstyle{jaune}=[draw=yellow!60!red,fill=yellow!60!red!30!white]
\tikzstyle{bleu}=[draw=blue,fill=blue!40!white]
\tikzstyle{gris}=[draw=black!80!white,fill=black!40!white]
\tikzstyle{rougef}=[draw=red,fill=red!60!white]
\tikzstyle{vertf}=[draw=green!80!black,fill=green!80!black!60!white]
\tikzstyle{jaunef}=[draw=yellow!80!black,fill=yellow!80!black!60!white]
\tikzstyle{bleuf}=[draw=blue,fill=blue!70!white]
\tikzstyle{rjaune}=[style=rond,style=jaune]
\tikzstyle{rbleu}=[style=rond,style=bleu]
\tikzstyle{rvert}=[style=rond,style=vert]
\tikzstyle{rrouge}=[style=rond,style=rouge]
\tikzstyle{rgris}=[style=rond,style=gris]
\tikzstyle{cjaune}=[style=carre,style=jaune]
\tikzstyle{cbleu}=[style=carre,style=bleu]
\tikzstyle{cvert}=[style=carre,style=vert]
\tikzstyle{crouge}=[style=carre,style=rouge]
\tikzstyle{cgris}=[style=carre,style=gris]
\tikzstyle{rjaunef}=[style=rond,style=jaunef]
\tikzstyle{rbleuf}=[style=rond,style=bleuf]
\tikzstyle{rvertf}=[style=rond,style=vertf]
\tikzstyle{rrougef}=[style=rond,style=rougef]%

\usetikzlibrary{shadows}
\tikzset{
diagonal fill/.style 2 args={fill=#2, path picture={
\fill[#1, sharp corners] (path picture bounding box.south west) -|
                         (path picture bounding box.north east) -- cycle;}},
reversed diagonal fill/.style 2 args={fill=#2, path picture={
\fill[#1, sharp corners] (path picture bounding box.north west) |-
                         (path picture bounding box.south east) -- cycle;}}
}

\usetikzlibrary{hobby,backgrounds,calc,trees}

\usepackage{csquotes}
\usepackage{cite}
\makeatletter\mathchardef\citemidpenalty=\@M\makeatother
\usepackage[most]{tcolorbox}

\newcommand{\complexityclass}[1]{\textnormal{\textsc{#1}}}

\newcommand{\AExpTime}{\complexityclass{AExpTime}}
\newcommand{\ExpSpace}{\complexityclass{ExpSpace}}

\newcommand{\NExpTime}{\complexityclass{NExpTime}}

\newcommand{\TwoNExpTime}{\complexityclass{2NExpTime}}

\newcommand{\coNExpTime}{\textnormal{co}\NExpTime}
\newcommand{\Vass}{\complexityclass{Vass}}
\newcommand{\deff}{\coloneqq}

\newcommand{\Q}{\mathbb{Q}}
\newcommand{\set}[1]{\{#1\}}
\renewcommand{\iff}{\leftrightarrow}
\newcommand{\cupdot}{\mathbin{\mathaccent\cdot\cup}}
\newcommand{\lesslex}{<_{\mathrm{lex}}}

\newcommand{\restr}{\!\!\restriction\!\!}

\newcommand{\lesseq}{\preceq}
\newcommand{\less}{\prec}

\makeatletter
\newcommand{\hathat}[1]{%
\begingroup%
  \let\macc@kerna\z@%
  \let\macc@kernb\z@%
  \let\macc@nucleus\@empty%
  \hat{\mathchoice%
    {\raisebox{.2ex}{\vphantom{\ensuremath{\displaystyle #1}}}}%
    {\raisebox{.2ex}{\vphantom{\ensuremath{\textstyle #1}}}}%
    {\raisebox{.16ex}{\vphantom{\ensuremath{\scriptstyle #1}}}}%
    {\raisebox{.14ex}{\vphantom{\ensuremath{\scriptscriptstyle #1}}}}%
    \smash{\hat{#1}}}%
\endgroup%
}
\makeatother

\newcommand{\Logic}[1]{\ensuremath{\mathrm{#1}}}
\newcommand{\FO}{\Logic{FO}}
\newcommand{\FOt}{\Logic{FO}^2}

\newcommand{\str}[1]{\mathfrak{#1}}
\def\mystr#1{\expandafter\def\csname str#1\endcsname{\str{#1}}}
\def\mystrall#1{\ifx#1\mystrall\else\mystr#1\expandafter\mystrall\fi}
\mystrall ABCDEFGHIJKLMNOPQRSTUVWXYZ\mystrall

\newcommand{\dom}[1]{\mathrm{#1}}
\def\mydom#1{\expandafter\def\csname dom#1\endcsname{\dom{#1}}}
\def\mydomall#1{\ifx#1\mydomall\else\mydom#1\expandafter\mydomall\fi}
\mydomall ABCDEFGHIJKLMNOPQRSTUVWXYZ\mydomall

\def\myrm#1{\expandafter\def\csname rm#1\endcsname{\mathrm{#1}}}
\def\myrmall#1{\ifx#1\myrmall\else\myrm#1\expandafter\myrmall\fi}
\myrmall ABCDEFGHIJKLMNOPQRSTUVWXYZ\myrmall

\newcommand{\pred}[1]{\mathrm{#1}}
\def\mypred#1{\expandafter\def\csname p#1\endcsname{\pred{#1}}}
\def\mypredall#1{\ifx#1\mypredall\else\mypred#1\expandafter\mypredall\fi}
\mypredall ABCDEFGHIJKLMNOPQRSTUVWXYZabcdefghjklmnopqrstuvwxyz \mypredall

\newcommand{\el}[1]{\mathit{#1}}
\def\myel#1{\expandafter\def\csname el#1\endcsname{\el{#1}}}
\def\myelall#1{\ifx#1\myelall\else\myel#1\expandafter\myelall\fi}
\myelall abcdefghijkmnopqrstuvwxyz\myelall

\newcommand{\aut}[1]{\mathdutchcal{#1}}
\def\myaut#1{\expandafter\def\csname aut#1\endcsname{\aut{#1}}}
\def\myautall#1{\ifx#1\myautall\else\myaut#1\expandafter\myautall\fi}
\myautall ABCDEFGHIJKLMNOPQRSTUVWXYZ\myautall

\newcommand{\homo}[1]{\mathfrak{#1}}

\newcommand{\homoh}{\homo{h}}

\newcommand{\desc}{\mathrel{\downarrow^{+}}}
\newcommand{\sigenc}{\sigma_{\mathrm{enc}}}
\newcommand{\sigaux}{\sigma_{\mathrm{aux}}}
\newcommand{\phitree}{\varphi_{\mathrm{tree}}}
\newcommand{\fNext}{\mathrm{Next}}
\newcommand{\fDown}{\mathrm{Down}}

\nolinenumbers
\begin{document}
\maketitle
\setcounter{page}{0}
\begin{abstract}
We study the two-variable fragment~$\FOt$ of first-order logic over
two-dimensional structures.
Formul{\ae} may employ two linear orders~$\lesseq_1, \lesseq_2$, their
successor relations~$\pS_1, \pS_2$, and auxiliary predicates of any arity.
Despite two decades of work, three cases of finite satisfiability in this
setting were still open; we settle all three.
First, finite satisfiability of~$\FOt$ with one linear order and one successor
relation, of the same or of a different order, is $\NExpTime$-complete, even
with arbitrary auxiliary predicates and constant symbols.
This closes a gap left open by Manuel and Zeume (CSL~2013).
As an application, a reduction based on swaps of adjacent elements shows that
order-invariance and order-successor-invariance of~$\FOt$ are
$\coNExpTime$-complete.
This reproves the main result of Zeume and Harwath (LICS~2016) and settles
their question on order-successor-invariance.
Second, finite satisfiability of~$\FOt$ with two linear orders is
$\ExpSpace$-hard, already without auxiliary non-unary predicates, matching the
upper bound known since 2010.
Our reduction combines Kiero\'{n}ski's encoding of alternating Turing machines
with the classical fact that the pre-order and post-order traversals of a tree
define its descendant~relation.
Third, with two linear orders and the successor of one of them, finite
satisfiability is in~$\ExpSpace$ even with arbitrary auxiliary predicates,
improving the previous $\TwoNExpTime$ bound and confirming a conjecture of
Zeume and Harwath.
Together with known results, our theorems show that in two dimensions,
auxiliary (non-unary) predicates and constants come for free: with or without
them, the complexities~coincide.
All our proofs are simple.
\end{abstract}

\clearpage
\section{Introduction}\label{sec:introduction}

\noindent \textbf{Two-variable logic.}\,
The two-variable fragment~$\FOt$ of first-order logic is among the most
prominent decidable fragments of~$\FO$.
Mortimer showed that it has the finite model property~\cite{Mortimer75}, and
Gr\"{a}del, Kolaitis and Vardi sharpened this to models of exponential size,
making its (finite) satisfiability problem
$\NExpTime$-complete~\cite{GraedelKV97}.
It is also reasonably expressive: $\FOt$ embeds many modal and description logics; over finite
words it is as expressive as unary temporal logic and remains
$\NExpTime$-complete~\cite{EtessamiVW02}; and over finite trees it captures
navigational XPath~\cite{MarxR05}, while its satisfiability problem becomes
$\ExpSpace$-complete~\cite{BenaimBCKLMW16}.

\vspace{0.25em}
\noindent \textbf{Semantic constraints.}\,
This expressiveness has its limits, however.
With only two variables at hand, $\FOt$ cannot say that a relation is
transitive, let alone that it is an equivalence or a linear order.
This has prompted a long line of work on $\FOt$ over structures in which some
\emph{distinguished} binary symbols have a fixed meaning; see the survey by
Kiero\'{n}ski, Pratt-Hartmann and Tendera~\cite{KieronskiPT18}.
Finite satisfiability remains decidable with one or two equivalence
relations~\cite{KieronskiO12}\cite{KieronskiT09}, with one transitive
relation~\cite{PrattHartmann18}, with one~\cite{Otto01} or two linear
orders~\cite{SchwentickZ12}, and with the successor relations of two linear
orders~\cite{CharatonikW13}\cite{Figueira12}.
It becomes undecidable with three equivalence relations~\cite{KieronskiO12},
with three linear orders~\cite{Kieronski11}, and with two linear orders
together with their successors~\cite[Thm.~4]{Manuel10}.
Recent work by Fiuk, Kiero\'{n}ski and Michielini extends this line to nested
equivalences and preorders for hierarchically~organised~data~\cite{FiukKM25}.

\vspace{0.25em}
\noindent \textbf{Data words and the plane.}\,
Among all semantic constraints, linear orders have received particular
attention, much of it sparked by \emph{data
words}~\cite{BojanczykDMSS11}: words in which every position carries, besides
a letter, a datum from an infinite domain, modelling, \eg, attribute values of
XML documents or process identifiers in traces of concurrent systems.
As structures, data words come with a linear order on positions, possibly its
successor, and an equivalence relating positions with equal data; $\FOt$ over
them is decidable, but its complexity is tied to reachability in Petri
nets~\cite{BojanczykDMSS11}.
Ordering the data turns the equivalence into a total
preorder~\cite{SchwentickZ12}, and if all data are moreover distinct, into
a second linear order.
Such a structure is best seen as a finite set of points in the plane, sorted
along the $x$-axis by~$\lesseq_1$ and along the $y$-axis by~$\lesseq_2$ (see
Figure~\ref{fig:plane}): these are the \emph{two-dimensional structures} of
our title.
We write $\FOt[\rmO]$ for $\FOt$ over them, where $\rmO$ lists the
distinguished symbols available, \ie linear orders~$\lesseq_i$ and
successor relations~$\pS_i$ (see Section~\ref{sec:preliminaries}); \eg
$\FOt[\set{\lesseq_1, \lesseq_2}]$ speaks about two linear orders.

\vspace{0.25em}
\noindent \textbf{Auxiliary binary symbols.}\,
The results above come in two flavours.
In the first, besides the distinguished symbols a sentence may use unary
predicates only; in the second, it may use arbitrary auxiliary ones, of which
binary predicates are the interesting case.
Zeume and Harwath, who phrase the two flavours as the logics $\Logic{EMSO}^2$
and $\Logic{ESO}^2$, point out that the distinction is easily
overlooked~\cite[\S1]{ZeumeH16}.
It matters: in the plane, for instance, a binary symbol can link any two points, however
far apart they lie in both orders, and many techniques that cut and paste
fragments of a model, developed for the unary flavour, break down in its
presence.
For two linear orders, the price used to be exponential: the best known upper
bound with auxiliary binary symbols~\cite[Thm.~7]{ZeumeH16} was exponentially
worse than the one without them~\cite[Thm.~3.1]{SchwentickZ12}.
Toru\'{n}czyk and Zeume give a direct translation of finite satisfiability
with auxiliary binary symbols into finite satisfiability without
them~\cite[\S5.2, Prop.~5.11]{TorunczykZ22}, but the translation is doubly exponential, and thus too costly~here.
Table~\ref{tab:results} collects what is known about~$\FOt[\rmO]$ in both
flavours, and marks the new results of this paper~with~$\bigstar$.

\vspace{0.25em}
\noindent \textbf{Our motivation.}\,
As the table shows, two decades of work still leave three gaps in the
complexity of~$\FOt$ over two linear orders. We aim to close all three of
them in this paper.
\begin{itemize}\itemsep0em
\item The $\ExpSpace$ upper bound for \emph{two linear orders} without
  auxiliary binary symbols dates back sixteen years~\cite{SchwentickZ10}\cite[Thm.~3.1]{SchwentickZ12}
  and was later reproved with automata~\cite{LabaiKOV20},
  yet the only lower bound is the $\NExpTime$ one inherited from~$\FOt$.
  The known $\ExpSpace$-hardness proof~\cite[Thm.~4.1]{SchwentickZ12}
  relies on a preorder, whose classes may be exponentially~large.
\item For \emph{one linear order and one foreign successor}, \ie
  $\FOt[\set{\lesseq_1, \pS_2}]$, the gap between $\NExpTime$ and $\ExpSpace$
  is listed as open by Manuel and Zeume~\cite[Fig.~3]{ManuelZ13}.
\item With \emph{auxiliary binary symbols}, only a $\TwoNExpTime$ upper bound
  is known for all of the above~\cite[Thm.~7]{ZeumeH16}, while Zeume and
  Harwath conjecture~$\ExpSpace$~\cite[\S4]{ZeumeH16}.
\end{itemize}

\vspace{0.25em}
\noindent \textbf{Order-invariance.}\,
Two linear orders on one domain also arise naturally in finite model theory,
where \emph{order-invariance} is an important property.
A sentence is order-invariant if it may use a linear order, but its truth
does not depend on which linear order is chosen.
Such sentences are ubiquitous: capturing results of descriptive complexity
assume an ordered structure, and the sentences they construct typically use
the order as an auxiliary device only~\cite[\S1]{ZeumeH16}.
Order-invariance is undecidable for~$\FO$~\cite[Exer.~9.3]{Libkin04}.
For~$\FOt$, Zeume and Harwath observed that a sentence fails to be
order-invariant exactly when some finite structure with two linear orders
satisfies it under the first order and violates it under the second.
Order-invariance of~$\FOt$ thus reduces to the complement of the finite
satisfiability problem for~$\FOt$ over two linear
orders~\cite[\S1]{ZeumeH16}.
With a successor at hand, however, the same idea leads to two linear orders
together with their successors, where satisfiability is
undecidable~\cite[Thm.~4]{Manuel10}.
Decidability of order-successor-invariance was thus left
open~\cite[\S7,~Question~3]{ZeumeH16} and, to the best of our knowledge,
has remained open ever since, despite subsequent progress on both the
expressive power and the complexity of order-invariant~$\FOt$
by Bednarczyk and Grange~\cite{BednarczykG25}.

\begin{table}[H]
\centering
\caption{Complexity of finite satisfiability for~$\FOt[\rmO]$, for signatures
whose auxiliary symbols are unary only (called $\Logic{EMSO}^2$
in~\cite{ZeumeH16}) or arbitrary (called~$\Logic{ESO}^2$).
Our results are marked by~$\bigstar$, the previous best bounds are shown in
grey, and lower bounds carry over from unary to arbitrary signatures; every
$\NExpTime$ lower bound is inherited from~$\FOt$~\cite{GraedelKV97}.
Cases symmetric to listed ones are omitted.
$\Vass$ is reachability in vector addition systems, Ackermann-complete~\cite{CzerwinskiO21}\cite{Leroux21} and equivalent to multicounter-automata emptiness~\cite{BojanczykDMSS11}.}
\label{tab:results}
\newcommand{\tabgroup}[1]{\multicolumn{3}{c}{\emph{#1}}}
\newcommand{\prev}[1]{\newline\textcolor{black!50}{was: #1}}
\newcommand{\ours}{~$\bigstar$}
\footnotesize
\renewcommand{\arraystretch}{1.15}
\begin{tabularx}{\textwidth}{@{}l
  >{\raggedright\arraybackslash}X
  >{\raggedright\arraybackslash}X@{}}
\toprule
Symbols $\rmO$ & Unary auxiliary symbols only
  & Arbitrary auxiliary symbols \\
\midrule
\tabgroup{One linear order or successor} \\
$\set{\pS}$
  & $\NExpTime$-complete~\cite{EtessamiVW02}
  & $\NExpTime$-complete~\cite[Thm.~3]{ZeumeH16} \\
$\set{\lesseq}$
  & $\NExpTime$-complete~\cite{EtessamiVW02}
  & $\NExpTime$-complete~\cite{Otto01} \\
$\set{\lesseq, \pS}$
  & $\NExpTime$-complete~\cite{EtessamiVW02}
  & $\NExpTime$-complete~\cite[Thm.~3]{ZeumeH16} \\
\midrule
\tabgroup{Two linear orders or successors} \\
$\set{\pS_1, \pS_2}$
  & $\NExpTime$-complete~\cite{CharatonikW13}\cite{Figueira12}
  & $\NExpTime$-complete~\cite{CharatonikW13} \\
$\set{\lesseq_1, \pS_2}$
  & $\NExpTime$-complete\ours{}
    \prev{in $\ExpSpace$~\cite[Thm.~3.1]{SchwentickZ12}}
  & $\NExpTime$-complete\ours{}
    \prev{in $\TwoNExpTime$~\cite[Thm.~7]{ZeumeH16}} \\
$\set{\lesseq_1, \lesseq_2}$
  & $\ExpSpace$-complete\ours{}
    \prev{in $\ExpSpace$~\cite[Thm.~3.1]{SchwentickZ12}\cite{LabaiKOV20}}
  & $\ExpSpace$-complete\ours{}
    \prev{in $\TwoNExpTime$~\cite[Thm.~7]{ZeumeH16}} \\
$\set{\lesseq_1, \pS_1, \lesseq_2}$
  & $\ExpSpace$-complete\ours{}
    \prev{in $\ExpSpace$~\cite[Thm.~3.1]{SchwentickZ12}\cite{LabaiKOV20}}
  & $\ExpSpace$-complete\ours{}
    \prev{in $\TwoNExpTime$~\cite[Thm.~7]{ZeumeH16}} \\
$\set{\pS_1, \lesseq_2, \pS_2}$
  & $\Vass$-complete~\cite[Cor.~11, Thm.~15]{ManuelZ13}
  & $\Vass$-complete~\cite[Thm.~9]{ZeumeH16} \\
$\set{\lesseq_1, \pS_1, \lesseq_2, \pS_2}$
  & undecidable~\cite[Thm.~4]{Manuel10}
  & undecidable~\cite[Thm.~4]{Manuel10} \\
\bottomrule
\end{tabularx}
\end{table}

\vspace{0.25em}
\noindent \textbf{Our results.}\,
We close all the gaps above, marked by~$\bigstar$ in
Table~\ref{tab:results}, and settle the question on order-successor-invariance
as well. In more detail, our contributions are as~follows.
\begin{enumerate}\itemsep0em
\item Finite satisfiability of $\FOt[\set{\lesseq_1, \pS_2}]$ is
  $\NExpTime$-complete, even with arbitrary auxiliary
  symbols~(Theorem~\ref{thm:one-order-one-succ}), which closes the gap
  listed by Manuel and Zeume~\cite[Fig.~3]{ManuelZ13}.
  The proof establishes an exponential-size model property by adapting the
  circular witnessing scheme of Gr\"{a}del, Kolaitis and
  Vardi~\cite{GraedelKV97} to the plane: we first confine all remote witnesses
  to a small set of elements, and then cut the model down along~$\lesseq_2$.
  As a by-product, we reprove $\NExpTime$-completeness of
  $\FOt[\set{\lesseq, \pS}]$~\cite[Thm.~3]{ZeumeH16} in a slightly simpler way.
\item Finite satisfiability of $\FOt[\set{\lesseq_1, \lesseq_2}]$ is
  $\ExpSpace$-hard, already without auxiliary binary
  symbols~(Theorem~\ref{thm:expspace-hardness}); previously, only the
  $\NExpTime$ lower bound inherited from~$\FOt$ was~known.
  We reduce from alternating Turing machines working in exponential time~\cite{ChandraKS81}.
  The descendant relation of a tree is definable by two linear orders,
  namely its pre-order and post-order traversals; this classical fact goes
  back to Dietz~\cite{Dietz82}.
  Benaim~\etal\ proved $\ExpSpace$-hardness over finite trees with
  Kiero\'{n}ski's encoding of runs~\cite[Thm.~3.8]{BenaimBCKLMW16}\cite[Thm.~5]{Kieronski06},
  which we also use as a black box, after fixing a minor gap (Appendix~\ref{app:runs}).
  The point is that two linear orders need not form a tree, so we show that
  the encoding survives in every model where the defined relation is
  merely~transitive.
\item Finite satisfiability of $\FOt[\set{\lesseq_1, \pS_1, \lesseq_2}]$ is
  in~$\ExpSpace$, even with arbitrary auxiliary
  symbols~(Theorem~\ref{thm:expspace-easyness}).
  This exponentially improves the previous $\TwoNExpTime$ bound and
  confirms the conjecture of Zeume and Harwath~\cite[\S4]{ZeumeH16}.
  As they hint~\cite[\S4]{ZeumeH16}, we combine their approach with the one
  without auxiliary~binary~predicates~\cite{SchwentickZ12}.
\item In Theorem~\ref{thm:orderinv} we show that both order-invariance and
  order-successor-invariance of~$\FOt$ are $\coNExpTime$-complete.
  The proof is strikingly short: swaps of adjacent elements connect any two
  linear orders, and constants can name the swapped pair, so both problems
  reduce to finite satisfiability of~$\FOt[\set{\lesseq, \pS}]$.
  This simple reduction reproves the main result of Zeume and
  Harwath~\cite{ZeumeH16} with the optimal bound (first shown
  in~\cite[Thm.~3.7]{BednarczykG25}) and settles their question on
  order-successor-invariance~\cite[\S7,~Question~3]{ZeumeH16}.
\end{enumerate}
Together, items~2 and~3 make $\FOt[\set{\lesseq_1, \lesseq_2}]$ and
$\FOt[\set{\lesseq_1, \pS_1, \lesseq_2}]$ $\ExpSpace$-complete, with and
without auxiliary binary symbols; for the second logic, this closes another gap listed
in~\cite[Fig.~3]{ManuelZ13}.
In fact, the two columns of Table~\ref{tab:results} now coincide row by row:
\emph{in two dimensions, auxiliary binary symbols come for free}.
What remains open concerns more than two dimensions, such as three successor
relations~\cite[\S7,~Question~1]{ZeumeH16}.
None of our proofs needs heavy machinery, which should make them
easy~to~build~upon.

\vspace{0.25em}
\noindent \textbf{Organisation.}\, Section~\ref{sec:preliminaries} fixes the notation.
Section~\ref{sec:one-order-one-succ} proves our first result, which also
serves as a gentle introduction to model surgery in the plane, and
Section~\ref{sec:orderinv} applies it to order-invariance.
Sections~\ref{sec:lower-bound} and~\ref{sec:upper-bound} establish the
matching $\ExpSpace$ lower and upper bounds, and
Section~\ref{sec:conclusions} concludes and discusses~open~problems.

\section{Preliminaries}\label{sec:preliminaries}

\newcommand{\varx}{\mathit{x}}
\newcommand{\vary}{\mathit{y}}
\newcommand{\tp}[2]{\mathrm{tp}^{#1}(#2)}
\newcommand{\tpone}[1]{\boldsymbol{\alpha}_{#1}}
\newcommand{\tptwo}[1]{\boldsymbol{\beta}_{#1}}

\noindent We assume familiarity with first-order logic~$\FO$ and with
computability and complexity~\cite[Secs.~2.1--2.3]{Libkin04}.
$\FOt$ is its two-variable fragment, \ie $\FO$ with only
$\varx, \vary$, and $|\varphi|$ is the size of~$\varphi$.
Our signatures are finite and function-free, but may contain constants and
predicates of any arity.
Among their symbols we single out a set~$\rmO$ of \emph{distinguished} binary
ones: each~$\lesseq_i$ denotes a linear order and each~$\pS_i$ its successor
relation.
$\FOt[\rmO]$ collects the $\FOt$ sentences with distinguished symbols in~$\rmO$,
and $\varphi[{\lesseq}{/}{\lesseq_i}]$ replaces each~$\lesseq$ in~$\varphi$
by~$\lesseq_i$; indices are dropped if only one index occurs in~$\rmO$.

\vspace{0.25em}
\noindent \textbf{Structures and types.}\, Structures $\strA, \strB$ are written
in fraktur and their domains $\domA, \domB$ in Roman; all of them are
\emph{finite} and non-empty.
A formula~$\varphi$ is written $\varphi(\varx)$
if it contains a single free variable $\varx$,
and it is a \emph{sentence} if it has none.
We write $\strA \restr \domB$ for the substructure of~$\strA$ induced
by~$\domB \subseteq \domA$.
An atomic \emph{$1$-type} (resp. \emph{$2$-type}) over
a signature~$\sigma$ is a maximal
satisfiable set of $\sigma$-literals in~$\varx$ (resp. in $\varx, \vary$, containing $\varx \neq \vary$),
identified with its conjunction; those over the signature of~$\varphi$ are
collected in $\tpone{\varphi}$ and~$\tptwo{\varphi}$.
Note that the sizes of $\tpone{\varphi}$ and~$\tptwo{\varphi}$
are exponential~in~$|\varphi|$.
Every $\ela \in \domA$ \emph{realises} a unique $1$-type~$\tp{\strA}{\ela}$ and
every pair of distinct elements $\ela, \elb$ realises a unique $2$-type~$\tp{\strA}{\ela,\elb}$.

\vspace{0.25em}
\noindent \textbf{Ordered structures and the plane.}\, A structure is
\emph{ordered} if it interprets every~${\lesseq_i} \in \rmO$ as a linear order
and every~${\pS_i} \in \rmO$ as the successor relation of some linear
order~$\lesseq_i$ on its domain, namely of~$\lesseq_i$ itself whenever
${\lesseq_i} \in \rmO$; we identify symbols with their interpretations.
Distinct elements $\ela, \elb$ are~$\lesseq_i$-\emph{close} if
$\pS_i(\ela,\elb)$ or $\pS_i(\elb,\ela)$, and $\lesseq_i$-\emph{remote}
otherwise; they are
\emph{close} if they are $\lesseq_i$-close for some~$i$ with ${\pS_i} \in \rmO$, and \emph{remote}
otherwise.
If $\rmO$ contains~$\pS_i$ but not~$\lesseq_i$, we speak of $\pS_i$-closeness.
An ordered structure with two linear orders is visualised as a set of
points in the plane, with $\lesseq_1$ ordering them along the $x$-axis and
$\lesseq_2$ along the $y$-axis.
Thus $\ela$ is to the \emph{left} (resp. \emph{right}) of~$\elb$ if
$\ela \less_1 \elb$ (resp. $\elb \less_1 \ela$) and \emph{below}
(resp. \emph{above}) it if $\ela \less_2 \elb$ (resp. $\elb \less_2 \ela$),
whence \emph{leftmost}, \emph{rightmost}, \emph{topmost} and \emph{bottommost}
elements.
Further, $\elb$ sits \emph{directly above}~$\ela$ if $\ela \less_2 \elb$ and no
element lies strictly between them, and a $2$-type \emph{places $\vary$ to the
right} (resp. \emph{left}) \emph{of~$\varx$} if it contains $\varx \lesseq_1 \vary$
(resp. $\vary \lesseq_1 \varx$). See~Figure~\ref{fig:plane}.

\tikzset{
  plane/.style ={x=5.8mm, y=5.8mm},
  ptl/.style   ={circle, draw=black!70, line width=0.45pt, minimum size=4.1mm,
                 inner sep=0pt, font=\scriptsize},
  succone/.style ={draw=bleu, line width=0.85pt, line cap=round,
                   -{Stealth[length=1.4mm,width=1.1mm]}},
  succtwo/.style ={draw=vert, line width=0.85pt, line cap=round,
                   dash pattern=on 1.5pt off 1.3pt,
                   -{Stealth[length=1.4mm,width=1.1mm]}},
  reledge/.style ={draw=rouge, line width=0.85pt, line cap=round,
                   -{Stealth[length=1.4mm,width=1.1mm]}},
  ranklbl/.style ={font=\tiny, text=black!45, inner sep=1pt}}
\newcommand{\figplanehead}[1]{\textbf{#1}}
\newcommand{\figplanerule}{\par\vspace{0.35em}%
  {\color{black!20}\hrule height 0.4pt}\vspace{0.45em}}
\newcommand{\figplaneremote}[1]{\textcolor{black!45}{#1}}

\begin{figure}[H]
\centering
\begin{minipage}[c]{0.39\textwidth}
\centering
\begin{tikzpicture}[plane]
  \draw[black!15, line width=0.35pt, step=1] (0.5,0.5) grid (5.5,5.5);
  \draw[line width=0.8pt, -{Stealth[length=2mm,width=1.6mm]}]
        (0.35,0.35) -- (5.95,0.35)
        node[font=\footnotesize, anchor=west, inner sep=2pt] {$\lesseq_1, \pS_1$};
  \draw[line width=0.8pt, -{Stealth[length=2mm,width=1.6mm]}]
        (0.35,0.35) -- (0.35,5.95)
        node[font=\footnotesize, anchor=south, inner sep=2pt] {$\lesseq_2, \pS_2$};
  \foreach \i in {1,...,5} {
    \node[ranklbl, anchor=north] at (\i,0.3) {\i};
    \node[ranklbl, anchor=east]  at (0.3,\i) {\i};
  }
  \node[ptl, fill=black!15] (a) at (1,1) {$\ela$};
  \node[ptl, fill=white]    (b) at (2,2) {$\elb$};
  \node[ptl, fill=white]    (c) at (3,4) {$\elc$};
  \node[ptl, fill=black!15] (d) at (4,3) {$\eld$};
  \node[ptl, fill=white]    (e) at (5,5) {$\ele$};
  \draw[succone] (a) to[bend left=14]  (b);
  \draw[succone] (b) to[bend left=9]   (c);
  \draw[succone] (c) to[bend left=13]  (d);
  \draw[succone] (d) to[bend left=9]   (e);
  \draw[succtwo] (a) to[bend right=14] (b);
  \draw[succtwo] (b) to[bend right=9]  (d);
  \draw[succtwo] (d) to[bend left=13]  (c);
  \draw[succtwo] (c) to[bend right=9]  (e);
  \draw[reledge] (a) to[bend left=30]  (c);
  \draw[reledge] (a) to[bend right=20] (d);
\end{tikzpicture}\\[0.5em]
{\footnotesize
\begin{tabular}{@{}l@{\;}l@{\hspace{1.4em}}l@{\;}l@{}}
\tikz[baseline=-0.5ex]{\draw[succone](0,0)--(0.5,0);} & $\pS_1$ &
\tikz[baseline=-0.5ex]{\draw[reledge](0,0)--(0.5,0);} & $\pR$ \\
\tikz[baseline=-0.5ex]{\draw[succtwo](0,0)--(0.5,0);} & $\pS_2$ &
\tikz[baseline=-0.5ex]{\node[ptl,fill=black!15,minimum size=2.6mm]{};} & $\pP$~holds
\end{tabular}}
\end{minipage}%
\hfill
\begin{minipage}[c]{0.59\textwidth}
\footnotesize \raggedright
\renewcommand{\arraystretch}{1.1}%
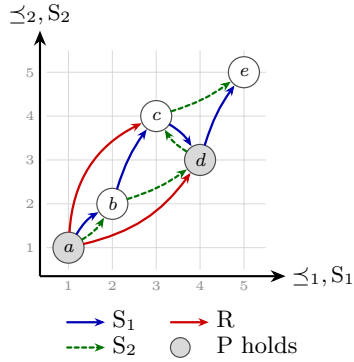
\figplanehead{Ranks.}\, Positions in the plane give the ranks:\\[0.25em]
\hspace*{0.8em}$\begin{array}{@{}l@{\quad}*{4}{c@{\;}c@{\;}}c@{}}
  \text{along } {\lesseq_1}{:} &
  \ela & \less_1 & \elb & \less_1 & \elc & \less_1 & \eld & \less_1 & \ele \\
  \text{along } {\lesseq_2}{:} &
  \ela & \less_2 & \elb & \less_2 & \eld & \less_2 & \elc & \less_2 & \ele
\end{array}$
\figplanerule
\figplanehead{Closeness.}\, Only $\elb$ is close to~$\ela$:\\[0.25em]
\hspace*{0.8em}\begin{tabular}{@{}l@{\hspace{1.2em}}c@{\hspace{1.2em}}c@{}}
                            & $\lesseq_1$ & $\lesseq_2$ \\[0.1em]
  $\ela, \elb$              & close  & close  \\
  $\elb, \elc$              & close  & \figplaneremote{remote} \\
  $\elb, \eld$              & \figplaneremote{remote} & close  \\
  $\ela$ and $\elc,\eld,\ele$ & \figplaneremote{remote} & \figplaneremote{remote}
\end{tabular}

\end{minipage}
\caption{A visualisation of an ordered structure~$\strA$ over
$\sigma \deff \set{\pP, \pR} \cup \rmO$ with
$\rmO \deff \set{\lesseq_1, \pS_1, \lesseq_2, \pS_2}$.}
\label{fig:plane}
\end{figure}

\vspace{-0.75em}

\vspace{0.25em}
\noindent \textbf{Decision problems.}\, The \emph{finite satisfiability problem}
for~$\FOt[\rmO]$ asks whether a given $\FOt[\rmO]$-sentence~$\varphi$ has a finite
ordered model, \ie a finite model in which every~${\lesseq_i} \in \rmO$ is interpreted as a
linear order and every~${\pS_i} \in \rmO$ as the successor relation of a linear order, namely
of~$\lesseq_i$ whenever ${\lesseq_i} \in \rmO$.
The remaining symbols of~$\varphi$, called \emph{auxiliary}, are interpreted arbitrarily.
We distinguish the case in which all auxiliary symbols are unary (the logic~$\Logic{EMSO}^2$
of~\cite{ZeumeH16}) from the general one (the logic~$\Logic{ESO}^2$); see
Table~\ref{tab:results}.

\vspace{0.25em}
\noindent \textbf{Scott normal form.}\, Passing to a normal form lets us reduce
the nesting of quantifiers and dispose of constants and of predicates of high
arity. An $\FOt[\rmO]$-sentence is in
\emph{Scott normal form} if its signature is constant-free and carries only
predicates of arity~$1$ or~$2$, and it~reads
\[
\forall{\varx}\forall{\vary}\, \chi \land \bigwedge_{i \leqslant m} \forall \varx \exists \vary\,
(\varx \neq \vary \land \chi_i),
\]
where $\chi$ and the~$\chi_i$ are quantifier-free.
Without loss of generality, we may assume this shape:

\begin{fact}\label{fact:scott}
  There is a nondeterministic polynomial-time procedure that,
  given an $\FOt[\rmO]$-sentence~$\varphi$, outputs an
  $\FOt[\rmO]$-sentence~$\psi$ in Scott normal form and of polynomial size,
  such that $\varphi$ has a finite ordered model if and only if some run
  outputs a~$\psi$~that~does.
\end{fact}

\noindent The proof is routine, so we omit it~\cite[Lems.~3.3--3.4, pp.~69--74]{PrattHartmann23}.
The conjuncts $\varx \neq \vary$ preserve satisfiability only over domains of at
least two elements~\cite[p.~71]{PrattHartmann23}, so one-element models are
tested directly, and if one exists, a trivially satisfiable~$\psi$ is output.
Nondeterminism is needed only to find such a model and to get rid of equalities
between constant symbols, both of which are~guessed.
Writing $\psi$ as above, $\strA \models \psi$ amounts to three conditions:
(i)~$\strA$ is ordered; (ii)~every $1$-type realised in~$\strA$
entails~$\chi(\varx,\varx)$ and every realised $2$-type entails~$\chi$;
and (iii)~every $\ela \in \domA$ admits \emph{witnesses}
$\elw_1, \ldots, \elw_m \in \domA \setminus \set{\ela}$ with
$\strA \models \chi_i[\ela, \elw_i]$.
A~model is thus a family of types in which every element \emph{demands}
$m$~witnesses, one fulfilling each~$\chi_i$.

\section{One Order, One Foreign Successor: Small Models in the Plane}
\label{sec:one-order-one-succ}

\newcommand{\Wz}{\rmW_{0}}
\newcommand{\Wo}{\rmW_{1}}
\newcommand{\Wt}{\rmW_{2}}
\newcommand{\Rtype}{\rmR_{\alpha}}

\noindent Call a successor relation ${\pS_i} \in \rmO$ \emph{foreign} if the
linear order~$\lesseq_i$ inducing it is itself absent from~$\rmO$: a sentence
then sees $\pS_i$-closeness, but has no access to $\lesseq_i$~itself.
Pairing one linear order with one foreign successor gives the simplest genuinely
two-dimensional signature, and it is the one we study here.
So let $\rmO \deff \set{\lesseq_1, \pS_2}$, and let $\varphi$ be a sentence in
Scott normal form~(Fact~\ref{fact:scott}), with universal conjunct~$\chi$ and
existential conjuncts~$\chi_1, \ldots, \chi_m$, where $m \leqslant |\varphi|$.
In an ordered structure~$\strA$ we write $\lesseq_2$ for the linear order
on~$\domA$ that induces~$\pS_2$, unique as $\domA$ is finite.
Not being in the signature, $\lesseq_2$ is recorded by no type: it is determined
by~$\pS_2$, and $\varphi$ constrains it only through~$\pS_2$.
Closeness is therefore $\pS_2$-closeness: distinct elements $\ela, \elb$ are
close exactly when one of them sits directly above the~other.
In every model of~$\varphi$ we fix, for each element~$\ela$ and each
$i \leqslant m$, one witness of~$\ela$ for~$\chi_i$; from now on,
\emph{witnesses} are these fixed ones only, and each surgery below says how it
fixes them anew.
A witness of~$\ela$, which differs from~$\ela$ by the normal form, is
called \emph{close} if it is close to~$\ela$, and \emph{remote}~otherwise.
Finally, we abbreviate $\rmN_1 \deff |\tpone{\varphi}|$ and
$\rmN_2 \deff |\tptwo{\varphi}|$, both exponential~in~$|\varphi|$.

\begin{theorem}\label{thm:one-order-one-succ}
  The finite satisfiability problem for $\FOt[\set{\lesseq_1, \pS_2}]$ is
  $\NExpTime$-complete, even with arbitrary-arity auxiliary predicates and
  constants.
\end{theorem}

\noindent This closes the gap listed by Manuel and
Zeume~\cite[Fig.~3]{ManuelZ13}.
Hardness comes from~$\FOt$; the known upper bounds were $\ExpSpace$ without
extra binary symbols~\cite[Thm.~3.1]{SchwentickZ12} and $\TwoNExpTime$ in
general~\cite[Thm.~7]{ZeumeH16}.
Ours utilises an exponential-model~property.

\begin{lemma}\label{lem:exp-model}
  If $\varphi$ has a finite model, then it has one of
  size~$2^{O(|\varphi|)}$.
\end{lemma}

\tikzset{
  reb/.style   ={x=4.6mm, y=4.0mm},
  shr/.style   ={x=7.0mm, y=7.0mm},
  apt/.style   ={ptcarre, vert, minimum width=2.4mm, minimum height=2.4mm,
                 inner sep=0pt},
  ego/.style   ={ptrond, gris, minimum height=3.1mm, inner sep=0pt},
  grpt/.style  ={minirond, fill=black!12, minimum height=1.9mm},
  wpt/.style   ={ptcarre, bleu, minimum width=2.5mm, minimum height=2.5mm,
                 inner sep=0pt},
  gpt/.style   ={minirond, fill=black!8, minimum height=2.2mm},
  dpt/.style   ={minirond, fill=rouge!25, draw=rouge, minimum height=2.2mm},
  grp/.style   ={ellipse, draw=black!55, line width=0.5pt,
                 dash pattern=on 1.2pt off 1.2pt, inner sep=0pt},
  oldw/.style  ={draw=rouge, line width=0.75pt, line cap=round,
                 dash pattern=on 1.5pt off 1.3pt,
                 -{Stealth[length=1.4mm,width=1.1mm]}},
  neww/.style  ={draw=rouge, line width=0.95pt, line cap=round,
                 -{Stealth[length=1.6mm,width=1.3mm]}},
  layer/.style ={draw=black!50, line width=0.55pt, line cap=round,
                 -{Stealth[length=1.3mm,width=1mm]}},
  succ/.style  ={draw=vert, line width=0.8pt, line cap=round,
                 dash pattern=on 1.5pt off 1.3pt,
                 -{Stealth[length=1.4mm,width=1.1mm]}},
  band/.style  ={fill=jaune!22, draw=none},
  cut/.style   ={fill=rouge!10, draw=rouge!40, line width=0.4pt,
                 dash pattern=on 1pt off 1pt},
  lab/.style   ={font=\scriptsize, inner sep=1.5pt}}

\begin{figure}[H]
\begin{minipage}[c]{0.46\textwidth}
\centering
\begin{tikzpicture}[reb]
  \fill[band] (0.6,0.3) rectangle (2.0,7.0);
  \fill[band] (7.6,0.3) rectangle (9.8,7.0);
  \draw[line width=0.8pt, -{Stealth[length=2mm,width=1.6mm]}]
        (0.3,0.3) -- (10.4,0.3);
  \draw[line width=0.8pt, -{Stealth[length=2mm,width=1.6mm]}]
        (0.3,0.3) -- (0.3,7.5);
  \node[lab, anchor=north east] at (10.45,0.2) {$\lesseq_1$};
  \node[lab, anchor=south west] at (0.2,7.5)   {$\lesseq_2$};
  \node[lab] at (1.3,6.6) {\envert{$\Wz$}};
  \node[lab] at (8.7,6.6) {\envert{$\Wz$}};
  \node[grp, minimum width=1.12cm, minimum height=0.58cm] (g1) at (3.5,5.6) {};
  \node[grp, minimum width=1.02cm, minimum height=0.58cm] (g2) at (6.2,4.4) {};
  \foreach \p in {(2.95,5.55),(3.5,5.72),(4.02,5.48)} \node[grpt] at \p {};
  \foreach \p in {(5.72,4.35),(6.22,4.52),(6.70,4.28)} \node[grpt] at \p {};
  \node[lab, anchor=south] at (3.5,6.32) {$\Wo$};
  \node[lab, anchor=north] at (6.2,3.68) {$\Wt$};
  \node[apt] (s)  at (1.3,5.4) {};
  \node[apt]      at (1.2,1.4) {}; \node[apt] at (1.6,3.2) {};
  \node[apt] (al) at (2.9,1.1) {};
  \node[apt] (w)  at (5.3,1.2) {};
  \node[apt] (v)  at (8.3,2.6) {};
  \node[apt]      at (8.0,0.9) {}; \node[apt] at (9.3,4.4) {};
  \node[lab, anchor=west] at (3.15,1.1) {\envert{$\alpha$}};
  \draw[layer] (s) -- (g1);  \draw[layer] (g1) -- (g2);
  \node[ego] (a) at (3.4,2.2) {};
  \node[lab, anchor=east]  at (3.0,2.2)  {$\ela$};
  \node[lab, anchor=north] at (5.3,0.85) {$\elw_i$};
  \node[lab, anchor=south] at (8.3,3.0) {$\elv_{\elw_i}$};
  \draw[oldw] (a) to[bend right=10] (w);
  \draw[neww] (a) to[bend left=9]  (v);
\end{tikzpicture}
\end{minipage}\hfill
\begin{minipage}[c]{0.52\textwidth}
\centering
\begin{tikzpicture}[shr]
  \fill[cut] (2.55,-0.42) rectangle (5.45,0.32);
  \node[wpt] (b1) at (0,0) {};
  \node[gpt] (y1) at (1,0) {};
  \node[gpt] (a1) at (2,0) {};
  \node[dpt] (a2) at (3,0) {};
  \node[dpt] (a3) at (4,0) {};
  \node[dpt] (c1) at (5,0) {};
  \node[gpt] (c2) at (6,0) {};
  \node[gpt] (c3) at (7,0) {};
  \node[gpt] (y2) at (8,0) {};
  \node[wpt] (b2) at (9,0) {};
  \foreach \s/\t in {b1/y1,y1/a1,a1/a2,a2/a3,a3/c1,c1/c2,c2/c3,c3/y2,y2/b2}
    \draw[succ] (\s) -- (\t);
  \draw[neww] (a1) to[out=80,in=100,looseness=0.73] (c2);
  \draw[rouge, line width=1pt, line cap=round, opacity=0.65]
        (2.55,-0.42) -- (5.45,0.32)  (2.55,0.32) -- (5.45,-0.42);
  \node[lab, anchor=south] at (0.5,0.18) {\envert{$\pS_2$}};
  \foreach \x/\t in {0/{\elb_1},2/{\ela_1},3/{\ela_2},
                     5/{\elc_1},6/{\elc_2},9/{\elb_2}}
    \node[lab, anchor=north] at (\x,-0.62) {$\t$};
  \foreach \x in {3,6}
    \node[lab, anchor=north, text=black!55] at (\x,-1.0) {$\succ_1\!\ela_1$};
  \draw[line width=0.8pt, -{Stealth[length=2mm,width=1.6mm]}]
        (0,-1.75) -- (9.5,-1.75);
  \node[lab, anchor=north east] at (9.5,-1.85) {$\lesseq_2$};
\end{tikzpicture}
\end{minipage}
\caption{Left: rebinding $\ela$'s remote witness~$\elw_i$ to $\elv_{\elw_i} \in \Wz$,
with $\Wo,\Wt$ closing it under witnesses. Right: cutting out $[\ela_2,\elc_2)$
between $\pS_2$-pairs $(\ela_1,\ela_2)$ and $(\elc_1,\elc_2)$ of equal
$2$-type ($\elb_1,\elb_2 \in \rmW$); the cut is legal as $\ela_2$ and~$\elc_2$
lie on the same $\lesseq_1$-side of~$\ela_1$, \ie $(\dagger)$~holds.}
\label{fig:warmup}
\end{figure}
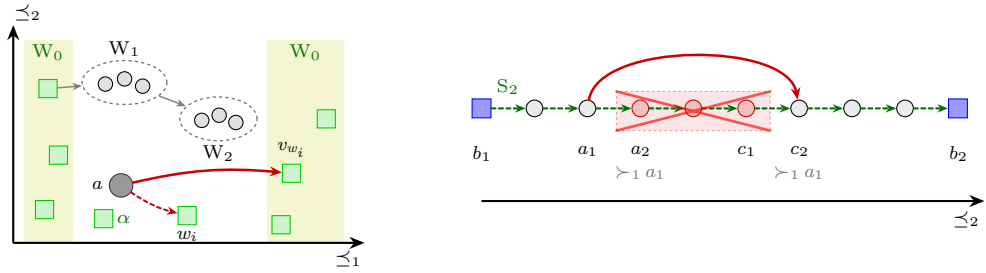

\vspace{-0.75em}

\noindent We adapt the witnessing scheme of Gr\"{a}del, Kolaitis and Vardi~\cite{GraedelKV97},
as used for two linear orders in~\cite[Lem.~3.6]{BednarczykG25}: we keep the extremal copies of
each $1$-type, the \emph{royal court}, and two spare ones, as $\pS_2$ makes at most two elements
close to a given one.
The $2$-type of a remote pair has no $\pS_2$-atom and may be overwritten
without disturbing the orders; that of a close pair may not.
We proceed in two steps, shown in Figure~\ref{fig:warmup}.
First we \emph{confine} all remote witnesses to an exponentially large
set~$\rmW$, rewiring only non-distinguished relations~(Lemma~\ref{lem:confine}).
Then we \emph{shrink} the model until only exponentially many elements lie
strictly between any two $\lesseq_2$-consecutive members
of~$\rmW$~(Lemma~\ref{lem:shrink}), which is possible
precisely because nothing outside~$\rmW$ is still needed as a remote witness.
The first step refines Otto's treatment of
$\FOt[\set{\lesseq}]$~\cite{Otto01}, the second follows the word-based
arguments~for~$\FOt[\set{\pS}]$~\cite{EtessamiVW02}.

\noindent We start with our first model surgery, described
in~Lemma~\ref{lem:confine}.

\begin{lemma}\label{lem:confine}
  Let $\strA \models \varphi$. There are a model $\strB \models \varphi$ and a
  set $\rmW \subseteq \domB$ of size $O(|\varphi|^3 \cdot \rmN_1)$ that
  contains the bottommost and the topmost element of~$\domB$, and every
  remote witness in~$\strB$.
  Moreover, $\domB = \domA$, and $\strA$ and~$\strB$ agree on all $1$-types and
  on $\lesseq_1$,~$\lesseq_2$~and~$\pS_2$.
\end{lemma}

\begin{proof}
\noindent We assemble the desired set $\rmW$ in two stages,
rewire connections concerning every element outside it,
and finally check that no witness was lost on the~way.

\vspace{0.25em}
\noindent \textbf{The royal court.}\, For each $1$-type~$\alpha$ realised
in~$\strA$, the court~$\Wz$ collects the $m{+}2$ leftmost and the
$m{+}2$ rightmost realisations of~$\alpha$, or all of them if there are fewer
than~$2m{+}4$.
It also contains the bottommost and the topmost element of~$\domA$, so that
$|\Wz| \leqslant 2(m{+}2)\rmN_1 + 2$.
Intuitively, $\Wz$ holds more copies of each $1$-type than any element needs: at
most $m$ witnesses are asked for, and two spares cover the close elements barred
from the surgery~below.

\vspace{0.25em}
\noindent \textbf{Closing under witnesses.}\, Let the set $\Wo$ hold the witnesses of
the elements of~$\Wz$ that lie outside~$\Wz$, and let $\Wt$ do the same
for~$\Wo$ outside $\Wz \cupdot \Wo$.
Put $\rmW \deff \Wz \cupdot \Wo \cupdot \Wt$.
Every element demands at most $m$ different witnesses,
hence $|\rmW_{i+1}| \leqslant m|\rmW_i|$ and
$|\rmW| \leqslant (1{+}m{+}m^2)|\Wz| = O(|\varphi|^3 \rmN_1)$.
In allocating witnesses we cannot simply stop at $\Wz$ and~$\Wo$.
The \emph{rebinding} below rewires pairs that enter~$\Wz$ from outside.
A pair from~$\Wo$ into~$\Wz$, however, may already be fixed as a witness pair of
its endpoint in~$\Wz$.
If $\Wo$ were rewired as well, such witnesses would be lost; we therefore spare
it and~keep~its~own~witnesses~as~$\Wt$.

\vspace{0.25em}
\noindent \textbf{Rebinding.}\, Fix $\ela \in \domA \setminus (\Wz \cup \Wo)$,
write $\elw_i$ for its witness for~$\chi_i$, and discard those $i$ for
which $\elw_i$ is close to~$\ela$.
Each remaining $\tp{\strA}{\ela,\elw_i}$ determines the $1$-type of~$\elw_i$ and
the $\lesseq_1$-side of~$\ela$ on which it lies; we split the indices into
classes according to this~pair.
Consider such a class~$\rmC$, say of $1$-type~$\alpha$ and to the right of~$\ela$.
If $\alpha$ has fewer than $2m{+}4$ realisations, all of them lie in~$\Wz$, the
$\elw_i$ with $i \in \rmC$ included, and we leave $\rmC$ alone.
Otherwise the $m{+}2$ rightmost realisations of~$\alpha$ form a
set $\Rtype \subseteq \Wz$, and for $i \in \rmC$ there are two cases: either
$\elw_i \in \Rtype$, or $\elw_i$ lies strictly to the left of all of~$\Rtype$
and then every element of~$\Rtype$ lies to the right of~$\ela$~as~well.
If the first case holds for all $i \in \rmC$, we leave $\rmC$ alone.
Otherwise all of~$\Rtype$ lies to the right of~$\ela$, and at least $m$ of
its elements are remote from~$\ela$, while the set
$\set{\elw_i \mid i \in \rmC}$ has at most $|\rmC| \leqslant m$ elements.
So we may pick an injection $\elw \mapsto \elv_{\elw}$ of this set into the
elements of~$\Rtype$ remote from~$\ela$ that fixes every $\elw \in \Rtype$,
redefine
\[
  \tp{\strB}{\ela, \elv_{\elw_i}} \deff \tp{\strA}{\ela, \elw_i}
  \qquad \text{for all } i \in \rmC ,
\]
and fix $\elv_{\elw_i}$ as the new witness of~$\ela$ for~$\chi_i$.
This is legitimate: the old and the new $2$-type agree on the $1$-types of both
arguments, both place $\vary$ to the right of~$\varx$, and neither contains a
$\pS_2$-atom.
So $\lesseq_1$, $\lesseq_2$, $\pS_2$ and all $1$-types survive, $\strB$ is again
ordered, and the imposed $2$-type is realised in~$\strA$ and hence
entails~$\chi$.

\vspace{0.25em}
\noindent \textbf{Nothing breaks.}\, Every pair touched above has one endpoint
outside $\Wz \cup \Wo$ and the other inside~$\Wz$.
The witnesses of the elements of~$\Wz$ lie in $\Wz \cup \Wo$, those of~$\Wo$
in~$\rmW$.
So no witness pair of an element of $\Wz \cup \Wo$ is ever touched,
and each such element keeps all its witnesses inside~$\rmW$.
Finally, take an element $\ela$ outside $\Wz \cup \Wo$.
Its close witnesses survive, since no element close to~$\ela$ is ever
picked as a target~$\elv_{\elw}$.
Its remote witnesses now lie in~$\Wz$ and still witness their demands: distinct
witnesses of one class get distinct targets, and distinct classes pick targets
from disjoint sets, of distinct $1$-types or, for one~$\alpha$ with at least
$2m{+}4$ realisations, its $m{+}2$ leftmost and $m{+}2$~rightmost~ones.
Together with $|\rmW| = O(|\varphi|^3\rmN_1)$ this proves~the~lemma.
\end{proof}

\noindent With all remote witnesses inside~$\rmW$, no element between two
$\lesseq_2$-consecutive members of~$\rmW$ is anybody's remote witness.
Such stretches can therefore be shortened by cutting between two
$\pS_2$-pairs of equal $2$-type, provided that the cut respects~$\lesseq_1$.

\begin{lemma}\label{lem:shrink}
  Let $\strA \models \varphi$ and let $\rmW \subseteq \domA$ contain the
  bottommost and the topmost element of~$\domA$, and every remote
  witness in~$\strA$.
  Then $\varphi$ has a model of size at most~$|\rmW| \cdot (\rmN_2{+}1)^2$.
\end{lemma}

\begin{proof}
As $\rmW$ contains the bottommost and the topmost element, every element
outside~$\rmW$ lies strictly between two $\lesseq_2$-consecutive members
of~$\rmW$.
Fix such members $\elb_1 \less_2 \elb_2$, let
$\elu_1 \less_2 \cdots \less_2 \elu_n$ be the elements strictly between them,
and suppose that $n \geqslant (\rmN_2{+}1)^2$.
For $1 \leqslant s < n$ put $\tau_s \deff \tp{\strA}{\elu_s,\elu_{s+1}}$.
We shall delete $\elu_{s+1}, \ldots, \elu_t$ for some $s < t$ and give
$(\elu_s, \elu_{t+1})$ the $2$-type~$\tau_s$.
This respects $1$-types if $\tau_s = \tau_t$, and it respects~$\lesseq_1$ if
\[
  \text{$\elu_{s+1}$ and~$\elu_{t+1}$ lie on the same side of~$\elu_s$
  (both to its left, or both to its right).}
  \tag{$\dagger$}
\]
If $\pS_2$ were the successor of~$\lesseq_1$, every $\elu_{s+1}$ would lie to
the right of~$\elu_s$.
Then all $s < t$ would satisfy~$(\dagger)$, and already
$n \geqslant \rmN_2{+}2$ would force $\tau_s = \tau_t$ for some $s < t$.
In general, the pairs $(\elu_s, \elu_{s+1})$ may take the shape
of a \enquote{ladder} whose \enquote{rungs} keep shifting to the left.
For example, it may happen that
$\elu_3 \less_1 \elu_4 \less_1 \elu_1 \less_1 \elu_2$ and $\tau_1 = \tau_3$;
then $\elu_2$ lies to the right of~$\elu_1$, but $\elu_4$ to its left.
So equal $2$-types do not guarantee~$(\dagger)$, but ladders turn out to be~short.

\vspace{0.25em}
\noindent \textbf{Finding a cut.}\, Suppose that no $s < t$ satisfies
$\tau_s = \tau_t$ and~$(\dagger)$.
Let $\beta$ be a $2$-type placing $\vary$ to the right of~$\varx$, and let
$t_1 < \cdots < t_r$ be the indices~$s$ with $\tau_s = \beta$.
For all $i < j$, as $\tau_{t_i} = \tau_{t_j}$,
we see that $(\dagger)$ fails for $t_i$ and~$t_j$.
Since $\elu_{t_i+1}$ lies to the right of~$\elu_{t_i}$, the element
$\elu_{t_j+1}$ lies to its left.
Thus, as in the example above, the pairs of $2$-type~$\beta$ form a ladder:
\[
  \elu_{t_j} \less_1 \elu_{t_j+1} \less_1 \elu_{t_i} \less_1 \elu_{t_i+1} .
  \tag{$\clubsuit$}
\]
Now fix $i < r$.
By~$(\clubsuit)$, $\elu_{t_{i+1}}$ lies to the left of~$\elu_{t_i}$ and
$\elu_{t_i+1}$ to its right, so $t_i{+}1 < t_{i+1}$.
Let $s_i < t_{i+1}$ be the largest index such that $\elu_{s_i}$ lies to the
right of~$\elu_{t_i}$; as $t_i{+}1$ qualifies, $s_i > t_i$.
Then $\elu_{s_i+1}$ lies to the left of~$\elu_{t_i}$, by the choice of~$s_i$,
or by~$(\clubsuit)$ if $s_i{+}1 = t_{i+1}$.
Hence
\[
  \elu_{s_i+1} \less_1 \elu_{t_i} \less_1 \elu_{s_i} ,
  \tag{$\spadesuit$}
\]
\ie $(\elu_{s_i}, \elu_{s_i+1})$ jumps leftwards over~$\elu_{t_i}$, and
$\tau_{s_i}$ places $\vary$ to the left of~$\varx$.
Moreover, $\tau_{s_i} \neq \tau_{s_j}$ for $i < j < r$.
Indeed, suppose otherwise.
As $s_i < t_{i+1} \leqslant t_j < s_j$ and $(\dagger)$ fails for $s_i$
and~$s_j$, the element $\elu_{s_j+1}$ lies to the right of~$\elu_{s_i}$.
By $(\spadesuit)$ and~$(\clubsuit)$, it then also lies to the right of
$\elu_{t_i}$ and of~$\elu_{t_j}$, which contradicts~$(\spadesuit)$~for~$j$.
Let $p$ and~$q$ be the numbers of $2$-types in~$\tptwo{\varphi}$ placing
$\vary$ to the right and to the left of~$\varx$, respectively.
So $\beta$ is taken by at most $q{+}1$ of the~$\tau_s$.
Symmetrically (reverse~$\lesseq_1$), every $2$-type placing $\vary$ to the
left of~$\varx$ is taken at most $p{+}1$ times.
As $p{+}q \leqslant \rmN_2$, counting the $n{-}1$ indices~$s$ by
$2$-types contradicts $n \geqslant (\rmN_2{+}1)^2$:
\[
  n - 1 \;\leqslant\; p(q{+}1) + q(p{+}1) \;\leqslant\;
  \tfrac{1}{2}\rmN_2^2 + \rmN_2 \;<\; (\rmN_2{+}1)^2 - 1 .
\]

\vspace{0.25em}
\noindent \textbf{Performing the cut.}\, So let $s < t$ satisfy
$\tau_s = \tau_t$ and~$(\dagger)$.
Put $\ela_1 \deff \elu_s$, $\ela_2 \deff \elu_{s+1}$, $\elc_1 \deff \elu_t$
and $\elc_2 \deff \elu_{t+1}$, where possibly $\ela_2 = \elc_1$.
Let $\strC$ arise from~$\strA$ by deleting every $\elu$ with
$\ela_2 \lesseq_2 \elu \less_2 \elc_2$ and setting
$\tp{\strC}{\ela_1, \elc_2} \deff \tau_s$, inheriting all else from~$\strA$.
Now $\elc_2$ sits directly above~$\ela_1$, and no other pair changes its
closeness.
The $2$-type $\tau_s$ contains $\pS_2(\varx,\vary)$, as $\elu_{s+1}$ sits
directly above~$\elu_s$.
It agrees with the $1$-types of $\ela_1$ and~$\elc_2$, as $\tau_s = \tau_t$.
It places $\vary$ to the right of~$\varx$ exactly when $\ela_2$ lies to the
right of~$\ela_1$, \ie by~$(\dagger)$, exactly when $\elc_2$ does.
So $\strC$~is~ordered.

\vspace{0.25em}
\noindent \textbf{Nothing breaks.}\, Every $1$-type and every $2$-type
realised in~$\strC$ is realised in~$\strA$, and thus entails~$\chi$.
Remote witnesses lie in~$\rmW$, which contains neither $\ela_1$ nor $\elc_2$ nor
any deleted element.
So they survive and serve the same demands.
Only $\ela_1$ and~$\elc_2$ can lose close witnesses.
Since $(\ela_1, \elc_2)$ now realises $\tau_s = \tau_t$, the element $\elc_2$
replaces~$\ela_2$ as a witness of~$\ela_1$, and $\ela_1$ replaces $\elc_1$ as a
witness of~$\elc_2$.
Thus $\strC \models \varphi$, and $\strC$ and~$\rmW$ meet the assumptions of
the lemma.
We repeat such cuts until fewer than $(\rmN_2{+}1)^2$ elements lie between any
two $\lesseq_2$-consecutive members of~$\rmW$.
Then at most
$|\rmW| + (|\rmW|{-}1)\big((\rmN_2{+}1)^2{-}1\big) \leqslant
|\rmW| \cdot (\rmN_2{+}1)^2$ elements~remain.
\end{proof}

\begin{proof}[Proof of Lemma~\ref{lem:exp-model}]
Let $\strA \models \varphi$ be finite.
Lemma~\ref{lem:confine} provides $\strB \models \varphi$ and a set~$\rmW$ of size
$O(|\varphi|^3\rmN_1)$ holding every remote witness and the bottommost and the
topmost element of~$\domB$.
Lemma~\ref{lem:shrink} then yields a model of size
$O\!\left(|\varphi|^3 \cdot \rmN_1 \cdot \rmN_2^2\right)$,
which~is~$2^{O(|\varphi|)}$.
\end{proof}

\noindent It remains to turn the exponential-model property into a decision
procedure.

\begin{proof}[Proof of Theorem~\ref{thm:one-order-one-succ}]
Given $\varphi \in \FOt[\set{\lesseq_1,\pS_2}]$, guess a sentence~$\psi$ in
Scott normal form output by a run of the procedure from
Fact~\ref{fact:scott}, then a structure~$\strA$ of
size~$2^{O(|\psi|)}$~(Lemma~\ref{lem:exp-model} for~$\psi$), and accept whenever
$\strA \models \psi$.
The last test is the three-item check from the preliminaries: $\strA$ is
ordered, every realised $1$-type and $2$-type entails the universal conjunct,
and every element has its $m$ witnesses.
This yields the desired $\NExpTime$ procedure~for~$\FOt[\set{\lesseq_1,\pS_2}]$.~\qedhere
\end{proof}

\noindent Finally, note that the axiom
$\forall{\varx}\forall{\vary}\,(\pS_2(\varx,\vary) \to \varx \lesseq_1 \vary)$
forces $\pS_2$ to be the successor of~$\lesseq_1$. As~a consequence,
Theorem~\ref{thm:one-order-one-succ} reproves $\NExpTime$-completeness
of~\mbox{$\FOt[\set{\lesseq,\pS}]$}~\cite[Thm.~3]{ZeumeH16}.

\section{Application: Deciding Order- and Order-Successor-Invariance}
\label{sec:orderinv}

\newcommand{\oinvFOt}{{\lesseq}\text{-inv}\,\FOt}
\newcommand{\osinvFOt}{{(\lesseq,\pS)}\text{-inv}\,\FOt}
\newcommand{\swapped}[1]{#1^{\ela\elb}}
\newcommand{\thle}{\theta_{\lesseq}}
\newcommand{\thS}{\theta_{\pS}}

\noindent A sentence $\varphi \in \FOt[\set{\lesseq, \pS}]$ is
\emph{order-successor-invariant} if, for every structure~$\strA$ without
distinguished relations, the truth of~$\varphi$ in the expansion
$(\strA, \lesseq, \pS)$ of~$\strA$ by a linear order~$\lesseq$ on~$\domA$ and
its successor~$\pS$ does not depend on the choice of~$\lesseq$.
If $\pS$ does not occur in~$\varphi$, we speak of \emph{order-invariance}.
The two classes are denoted $\osinvFOt$ and~$\oinvFOt$.
Recall that $\varphi$ is order-invariant if and only if $\varphi_1 \land \neg\varphi_2$
has no finite ordered model, where $\varphi_i$ is $\varphi$ with every~$\lesseq$
replaced by~$\lesseq_i$~\cite[\S1]{ZeumeH16}; with a successor, the analogous
formula lands in an undecidable fragment~\cite[Thm.~4]{Manuel10}.
The following theorem settles the decidability
of~$\osinvFOt$~\cite[\S7, Question~3]{ZeumeH16} and also
reproves the optimal complexity bound for~$\oinvFOt$~\cite[Thm.~3.7]{BednarczykG25}.

\begin{theorem}\label{thm:orderinv}
  Deciding whether an $\FOt[\set{\lesseq, \pS}]$-sentence is
  order-successor-invariant is $\coNExpTime$-complete.
  The same holds for order-invariance
  of~$\FOt[\set{\lesseq}]$-sentences.
\end{theorem}

\noindent The proof is short and elementary: we reprove the main result of
Zeume and Harwath with the optimal bound and answer their open question by a
simple reduction to finite satisfiability of~$\FOt$ with a single linear
order and~its~successor.
For elements $\ela, \elb$ with $\pS(\ela, \elb)$, let $\swapped{\lesseq}$ be $\lesseq$ with $\ela$
and~$\elb$ swapped, and let $\swapped{\pS}$ be its~successor.

\begin{observation}\label{obs:adjacent-swap}
  A sentence $\varphi$ is order-successor-invariant if and only if
  $(\strA, \lesseq, \pS) \models \varphi$ implies
  $(\strA, \swapped{\lesseq}, \swapped{\pS}) \models \varphi$ for all
  structures $\strA$, orders~$\lesseq$,~and elements~$\ela, \elb$~satisfying~$\pS(\ela, \elb)$.
\end{observation}

\noindent Indeed, bubble sort turns any linear order into any other by swaps
of~neighbours.
Such swaps are local: $\swapped{\lesseq}$ differs from~$\lesseq$ on the
pair~$\ela, \elb$ alone, and, writing $\ela^{-}$ for the predecessor of~$\ela$
and $\elb^{+}$ for the successor of~$\elb$ (if they exist), $\swapped{\pS}$
replaces the pairs $(\ela^{-}, \ela)$, $(\ela, \elb)$, $(\elb, \elb^{+})$
of~$\pS$ by $(\ela^{-}, \elb)$, $(\elb, \ela)$, $(\ela, \elb^{+})$.
Naming the elements $\ela, \elb$ by fresh constants $\pc, \pd$,
both relations become definable with the following quantifier-free~formul{\ae}:
\begin{align*}
  \hspace{-1.5em}\thle(\varx, \vary) & \deff
    \big(\varx \lesseq \vary \land \neg(\varx = \pc \land \vary = \pd)\big)
    \lor (\varx = \pd \land \vary = \pc), \\
  \hspace{-1.5em}\thS(\varx, \vary) & \deff
    \big(\pS(\varx, \vary) \land \varx, \vary \notin \set{\pc, \pd}\big)
    \lor \big(\pS(\varx, \pc) \land \vary = \pd\big) \lor (\varx = \pd \land \vary = \pc)
    \lor \big(\varx = \pc \land \pS(\pd, \vary)\big),
\end{align*}
where $\varx, \vary \notin \set{\pc, \pd}$ abbreviates four disequalities.
Replacing in~$\varphi$ every atom $t \lesseq t'$ by~$\thle(t, t')$ and every
$\pS(t, t')$ by~$\thS(t, t')$ yields an $\FOt$-sentence
$\varphi[{\lesseq}{/}\thle, \pS{/}\thS]$ of size~$O(|\varphi|)$. We prove:

\begin{proof}[Proof of Theorem~\ref{thm:orderinv}]
Let $\varphi \in \FOt[\set{\lesseq, \pS}]$ and let $\pc, \pd$ be fresh
constants.
In every ordered structure with $\pS(\pc, \pd)$, $\thle$ and~$\thS$ define
$\lesseq^{\pc\pd}$ and~$\pS^{\pc\pd}$.
So, by Observation~\ref{obs:adjacent-swap}, $\varphi$ is
order-successor-invariant if and only if
$\psi_{\varphi} \deff \pS(\pc, \pd) \land \varphi \land
\neg\varphi[{\lesseq}{/}\thle, \pS{/}\thS]$ is not finitely satisfiable.
As $\psi_{\varphi}$ has polynomial size and finite satisfiability of
$\FOt[\set{\lesseq, \pS}]$ is in~$\NExpTime$
(Theorem~\ref{thm:one-order-one-succ} and the observation closing
Section~\ref{sec:one-order-one-succ}, or~\cite[Thm.~3]{ZeumeH16}), the upper
bound follows.
For just order-invariance $\pS$ is unnecessary: replace $\pS(\pc, \pd)$
in~$\psi_{\varphi}$ by $\pc \less \pd \land \neg\exists\varx\,(\pc \less \varx \land
\varx \less \pd)$.
Hardness follows from that of order-invariance~\cite[Thm.~3.2]{BednarczykG25} on $\pS$-free~sentences.
\end{proof}

\noindent Nothing here is specific to two variables: the argument works for every logic with
equality and constants that is closed under Boolean connectives and under substituting
quantifier-free formul{\ae} for atoms. For instance, with~\cite[Thm.~6.20]{CharatonikW16} it
yields decidability of both kinds of invariance for the two-variable logic with
counting~$\Logic{C}^2$, as asked in~\cite[\S6]{BednarczykG25}.

\section{Configuration Trees: An $\ExpSpace$ Lower Bound}
\label{sec:lower-bound}

\noindent The complexity of finite satisfiability
for~$\FOt[\set{\lesseq_1, \lesseq_2}]$ has long been open.
Without auxiliary binary symbols, Schwentick and Zeume placed it
in~$\ExpSpace$~\cite[Thm.~3.1]{SchwentickZ12} and proved hardness for a
linear order and a total preorder~\cite[Thm.~4.1]{SchwentickZ12}.
For two linear orders, only $\NExpTime$-hardness from plain $\FOt$ was~known.
The following theorem closes this~gap. Unlike the setting of data words~\cite[\S5]{SchwentickZ12},
unary predicates may hold simultaneously~here.

\begin{theorem}\label{thm:expspace-hardness}
  The finite satisfiability problem for $\FOt[\set{\lesseq_1, \lesseq_2}]$ is
  $\ExpSpace$-hard, already over signatures without auxiliary non-unary~predicates and constant symbols.
\end{theorem}

\noindent \textbf{The reduction.}\, It suffices to encode accepting
runs of alternating Turing machines working in exponential time~\cite[Cor.~3.6]{ChandraKS81}.
We lay them out in \emph{configuration trees}, which carry a descendant
relation~$\desc$, and axiomatise them by Kiero\'{n}ski's
construction~\cite[Thm.~5]{Kieronski06}, used as a black box
(Lemma~\ref{lem:runs}; Appendix~\ref{app:runs} fixes a minor gap); Benaim~\etal\ used it in the same way to show
$\ExpSpace$-hardness over finite trees~\cite[Thm.~3.8]{BenaimBCKLMW16}.
Two facts make it work here.
First, $\desc$ is definable by two linear orders, the pre-order and the
post-order traversals of the tree~\cite{Dietz82}~(Lemma~\ref{lem:desc-orders}).
Second, two linear orders need not form a tree; the
encoding nevertheless survives in models in which the defined relation is merely
transitive~(Lemma~\ref{lem:runs}(iii)).

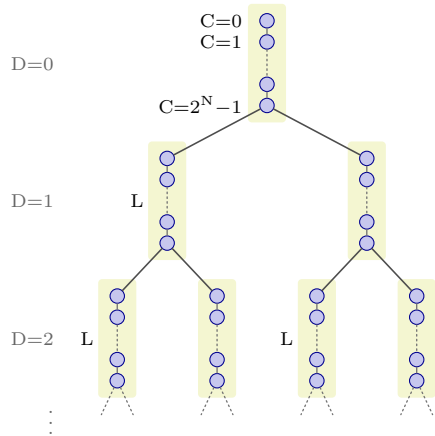
\begin{figure}[H]
\centering
\tikzset{
  ctree/.style   ={x=6mm, y=2.8mm},
  ct-cell/.style ={minirond, fill=bleu!22, draw=bleu!80!black,
                   minimum height=1.9mm},
  ct-conf/.style ={fill=jaune!22, draw=none, rounded corners=1.5pt},
  ct-in/.style   ={draw=black!55, line width=0.5pt, line cap=round},
  ct-br/.style   ={draw=black!70, line width=0.6pt, line cap=round},
  ct-dots/.style ={draw=black!55, line width=0.5pt, line cap=round,
                   dash pattern=on 0.6pt off 1.2pt},
  ct-lab/.style  ={font=\scriptsize, inner sep=1.5pt},
  ct-dep/.style  ={font=\scriptsize, text=black!60, inner sep=1.5pt}}
\newcommand{\ctconf}[3]{%
  \fill[ct-conf] (#2-0.42,#3+0.8) rectangle (#2+0.42,#3-4.8);
  \node[ct-cell] (#1-0) at (#2,#3)   {};
  \node[ct-cell] (#1-1) at (#2,#3-1) {};
  \node[ct-cell] (#1-2) at (#2,#3-3) {};
  \node[ct-cell] (#1-3) at (#2,#3-4) {};
  \draw[ct-in]   (#1-0) -- (#1-1);
  \draw[ct-dots] (#1-1) -- (#1-2);
  \draw[ct-in]   (#1-2) -- (#1-3);}
\begin{tikzpicture}[ctree]
  \ctconf{r}{0}{0}
  \ctconf{a}{-2.2}{-6.5}   \ctconf{b}{2.2}{-6.5}
  \ctconf{aa}{-3.3}{-13}   \ctconf{ab}{-1.1}{-13}
  \ctconf{ba}{1.1}{-13}    \ctconf{bb}{3.3}{-13}
  \foreach \p/\ch in {r/a, r/b, a/aa, a/ab, b/ba, b/bb}
    \draw[ct-br] (\p-3) -- (\ch-0);
  \foreach \ch in {aa, ab, ba, bb} {
    \draw[ct-dots] (\ch-3) -- ++(-0.35,-1.6);
    \draw[ct-dots] (\ch-3) -- ++( 0.35,-1.6);
  }
  \foreach \x/\y in {-2.2/-6.5, -3.3/-13, 1.1/-13}
    \node[ct-lab, anchor=east] at (\x-0.42,\y-2) {$\pL$};
  \node[ct-lab, anchor=east] at (-0.45, 0) {$\pC{=}0$};
  \node[ct-lab, anchor=east] at (-0.45,-1) {$\pC{=}1$};
  \node[ct-lab, anchor=east] at (-0.45,-4) {$\pC{=}2^{\rmN}{-}1$};
  \node[ct-dep, anchor=east] at (-4.6, -2)   {$\pD{=}0$};
  \node[ct-dep, anchor=east] at (-4.6, -8.5) {$\pD{=}1$};
  \node[ct-dep, anchor=east] at (-4.6,-15)   {$\pD{=}2$};
  \node[ct-dep, anchor=east] at (-4.6,-18.6) {$\vdots$};
\end{tikzpicture}
\caption{The $\rmN$-configuration tree up to depth~$2$: columns are configurations of $2^{\rmN}$~cells, lines form the Hasse diagram of~$\desc$ (read downwards), and $\pL$ marks left~configurations.}
\label{fig:configuration-tree}
\end{figure}

\noindent \textbf{Configuration trees.}\, Fix~$\rmN$.
The $\rmN$-configuration tree (Figure~\ref{fig:configuration-tree}) arises
from the full binary tree of height~$2^{\rmN}{-}1$ by replacing each node~$w$
with a \emph{configuration} of $2^{\rmN}$~\emph{cells}~$ww'$, where
$w' \in \set{0,1}^{\rmN}$.
The cell~$ww'$ has \emph{depth}~$|w|$ and \emph{position}~$p$, the number
written~$w'$ in binary; it is the \emph{first} cell of its configuration
if~$p = 0$, the \emph{last} one if~$p = 2^{\rmN}{-}1$, and, if~$p > 0$, the
\emph{next} cell after the one of position~$p{-}1$.
The configurations $w0$ and~$w1$ are the left and the right \emph{child}
of~$w$, and the \emph{root} configuration is the one with~$w$~empty.

\begin{definition}\label{def:configuration-tree}
  Let $\sigenc \deff \set{\pD_0, \ldots, \pD_{\rmN{-}1}, \pC_0, \ldots,
  \pC_{\rmN{-}1}, \pL, \desc}$, with~$\desc$ binary and the rest
  unary.
  The \emph{$\rmN$-configuration tree} is the $\sigenc$-structure~$\strA$
  with domain~$\domA$ of all $ww' \in \set{0,1}^*$ with $|w| < 2^{\rmN}$
  and~$|w'| = \rmN$, where, numbering bits from~$0$, most significant~first:
  \begin{itemize}\itemsep0em
    \item $ww' \in \pD_i^{\strA}$ if and only if the $i$-th bit of the depth
          of~$ww'$ is~$1$;
    \item $ww' \in \pC_i^{\strA}$ if and only if the $i$-th bit of the
          position of~$ww'$ is~$1$;
    \item $ww' \in \pL^{\strA}$ if and only if $w$ ends with~$0$;
    \item $vv' \desc ww'$ holds if and only if $v$ is a proper prefix of~$w$,
          or $v = w$ and~$v' \lesslex w'$.
  \end{itemize}
\end{definition}

\vspace{0.25em}
\noindent \textbf{Defining configuration trees.}\, Comparisons of the
depths~$\pD(\varx)$ and positions~$\pC(\varx)$ encoded at~$\varx$, such as
$\pD(\varx) = \ell$ or~$\pC(\vary) = \pC(\varx){+}1$, are expressible by
quantifier-free formul{\ae} of size polynomial in~$\rmN$; \eg
$\pD(\varx) = \ell$ is the conjunction of~$\pD_i(\varx)$ for the bits~$1$
of~$\ell$ and~$\neg\pD_i(\varx)$ for its~bits~$0$.
Adapting Kiero\'{n}ski's axioms~\cite[$\Phi_1$--$\Phi_3$, p.~1672]{Kieronski06},
we let $\phitree$ be the conjunction~of
\begin{align*}
  & \exists\varx\, \big( \pD(\varx) = 0 \land \pC(\varx) = 0 \land
    \neg\pL(\varx) \big), \\
  & \forall\varx\, \big( \pC(\varx) \neq 2^{\rmN}{-}1 \to \exists\vary\,
    ( \varx \desc \vary \land \pD(\vary) = \pD(\varx) \land
    \pC(\vary) = \pC(\varx){+}1 \land (\pL(\varx) \iff \pL(\vary)) ) \big), \\
  & \forall\varx\, \big( (\pC(\varx) = 2^{\rmN}{-}1 \land
    \pD(\varx) \neq 2^{\rmN}{-}1) \\
  & \qquad\qquad \to \textstyle\bigwedge_{b \in \set{0,1}}
    \exists\vary\, ( \varx \desc \vary \land \pD(\vary) = \pD(\varx){+}1 \land
    \pC(\vary) = 0 \land \mu_b(\vary) ) \big),
\end{align*}
where $\mu_0 \deff \pL$ and~$\mu_1 \deff \neg\pL$.
The conjuncts say that the first cell of the root configuration exists, that
every cell but the last of its configuration has a next cell, and that below
the last cell of every configuration of depth less than~$2^{\rmN}{-}1$ lie the
first cells of a left and of a right~child.
Clearly, $\strA \models \phitree$.
Models of~$\phitree$ need not be configuration trees, as witnesses may
coincide, but those interpreting~$\desc$ transitively contain a homomorphic
image of~$\strA$.
This suffices to encode runs in \emph{expansions} of~$\strA$, \ie in
structures obtained from~$\strA$ by interpreting further unary predicates.

\begin{lemma}[after~\cite{Kieronski06}]\label{lem:runs}
  Let $\autM$ be an alternating Turing machine working in time~$2^{n^k}$ on
  inputs of length~$n$.
  There are a constant~$c$, a set~$\sigaux$ of unary predicates, and a
  polynomial-time procedure computing, from an input~$u$ of length~$n$, an
  $\FOt$-sentence~$\varphi_{\autM,u}$ over~$\sigenc \cup \sigaux$ such that,
  for~$\rmN \deff n^k{+}c$, these are equivalent:
  (i)~$\autM$ accepts~$u$;
  (ii)~some expansion of~$\strA$ by~$\sigaux$
  satisfies~$\phitree \land \varphi_{\autM,u}$;
  (iii)~$\phitree \land \varphi_{\autM,u}$ has a finite model in
  which~$\desc$ is~transitive.
\end{lemma}

\vspace{0.25em}
\noindent \textbf{Descendants by two linear orders.}\, It remains to
define~$\desc$ by two linear orders, which the quantifier-free
formula~$\varx \less_1 \vary \land \vary \less_2 \varx$~does.

\begin{lemma}\label{lem:desc-orders}
  In every structure with linear orders~$\lesseq_1$ and~$\lesseq_2$, the
  formula $\varx \less_1 \vary \land \vary \less_2 \varx$ defines a transitive
  relation.
  Moreover, the domain of the $\rmN$-configuration tree~$\strA$ admits linear
  orders~$\lesseq_1$ and~$\lesseq_2$ such that $\ela \desc \elb$ holds
  in~$\strA$ if and only if $\ela \less_1 \elb$ and~$\elb \less_2 \ela$.
\end{lemma}

\begin{proof}
  The formula defines the intersection of~$\less_1$ with the converse
  of~$\less_2$, which is transitive.
  For the second part, observe that~$\desc^{\strA}$ is the proper-ancestor
  relation of the rooted tree on~$\domA$ in which the root is~$0^{\rmN}$, the
  only child of a cell other than the last of its configuration is its next
  cell, and the children of the last cell~$w1^{\rmN}$, for
  $|w| < 2^{\rmN}{-}1$, are $w0\,0^{\rmN}$ and~$w1\,0^{\rmN}$.
  Let $\lesseq_1$ and~$\lesseq_2$ be the pre-order and the post-order
  traversals of this tree, both visiting $w0\,0^{\rmN}$ before~$w1\,0^{\rmN}$.
  A proper ancestor of a cell precedes it in pre-order and follows it in
  post-order.
  Incomparable cells lie in the subtrees of distinct children of their lowest
  common ancestor, which both traversals visit one after the other, so
  $\lesseq_1$ and~$\lesseq_2$ order them alike.
  Hence $\ela \less_1 \elb$ and~$\elb \less_2 \ela$ hold if and only if~$\ela$
  is a proper ancestor of~$\elb$, \ie if and only if~$\ela \desc \elb$ holds.
\end{proof}

\begin{proof}[Proof of Theorem~\ref{thm:expspace-hardness}]
  As $\AExpTime = \ExpSpace$~\cite[Cor.~3.6]{ChandraKS81}, some alternating
  Turing machine~$\autM$ working in time~$2^{n^k}$ has an $\ExpSpace$-hard
  acceptance problem.
  Given an input~$u$, let $\psi_u$ arise from
  $\phitree \land \varphi_{\autM,u}$ (Lemma~\ref{lem:runs}) by replacing every
  atom $t \desc t'$, with $t, t' \in \set{\varx, \vary}$,
  by~$t \less_1 t' \land t' \less_2 t$.
  Then $\psi_u$ is an $\FOt[\set{\lesseq_1, \lesseq_2}]$-sentence whose other
  symbols are unary, and it is computable in polynomial time.
  If $\autM$ accepts~$u$, equipping the expansion of~$\strA$ from
  Lemma~\ref{lem:runs}(ii) with the orders of Lemma~\ref{lem:desc-orders}
  yields a finite ordered model of~$\psi_u$.
  Conversely, interpreting~$\desc$ by
  $\varx \less_1 \vary \land \vary \less_2 \varx$ in a finite ordered model
  of~$\psi_u$ yields a model as in Lemma~\ref{lem:runs}(iii), by
  Lemma~\ref{lem:desc-orders}; hence $\autM$ accepts~$u$.
\end{proof}

\section{Processing Lines: An $\ExpSpace$ Upper Bound}
\label{sec:upper-bound}

\newcommand{\wst}[2]{\mathrm{wst}^{#1}(#2)}
\newcommand{\pro}[2]{\mathrm{pl}^{#1}(#2)}
\newcommand{\bx}{\mathbf{x}}
\newcommand{\opp}[1]{\bar{#1}}

\noindent We now add the successor of one of the two linear orders.
Without auxiliary binary symbols, the resulting problem lies in~$\ExpSpace$~\cite[Thm.~3.1]{SchwentickZ12}, whereas with them only a $\TwoNExpTime$ upper bound was known~\cite[Thm.~7]{ZeumeH16}.
Zeume and Harwath conjectured that $\ExpSpace$ suffices~\cite[\S4]{ZeumeH16}, and we confirm~this.

\begin{theorem}\label{thm:expspace-easyness}
  The finite satisfiability problem for $\FOt[\set{\lesseq_1, \pS_1, \lesseq_2}]$
  is in~$\ExpSpace$, even over signatures with arbitrary auxiliary~symbols.
\end{theorem}

\noindent Swapping the two orders, we work with $\FOt[\set{\lesseq_1, \lesseq_2, \pS_2}]$ instead, where the successor belongs to the vertical order~$\lesseq_2$.
We follow the plane-sweep approach of Schwentick and Zeume~\cite[\S3.3]{SchwentickZ12}: a horizontal line sweeps a model bottom-up, and an algorithm guesses, line by line, a description of what the line sees, checking consecutive lines only.
With auxiliary binary symbols around, such a description has to be a substructure of the model, as in the profiles of Zeume and Harwath~\cite[proof of Lem.~8]{ZeumeH16}, but it must also keep every witness pending across the line.
The crux is to keep these substructures exponentially small (Lemma~\ref{lem:small-lines}) and to glue them back into a model (Lemma~\ref{lem:lines-imply-sat}).
So let $\varphi$ be an $\FOt[\set{\lesseq_1, \lesseq_2, \pS_2}]$-sentence in Scott normal form~(Fact~\ref{fact:scott}), with universal conjunct~$\chi$ and existential conjuncts $\chi_1, \ldots, \chi_m$.

\tikzset{
  pl/.style      ={x=5mm, y=4mm},
  pl-pt/.style   ={circle, draw=black!70, fill=white, line width=0.45pt,
                   minimum size=4.1mm, inner sep=0pt},
  pl-min/.style  ={pl-pt, draw=bleu!85!black, fill=bleu!18},
  pl-max/.style  ={pl-pt, draw=orange!80!black, fill=orange!22},
  pl-ghost/.style={pl-pt, draw=black!35, fill=none, text=black!45,
                   dash pattern=on 1.2pt off 1pt},
  pl-below/.style={fill=black!5, draw=none},
  pl-cut/.style  ={draw=black!60, line width=0.6pt,
                   dash pattern=on 2.5pt off 1.5pt},
  pl-rel/.style  ={draw=rouge, line width=0.85pt, line cap=round,
                   -{Stealth[length=1.4mm,width=1.1mm]}},
  pl-lab/.style  ={font=\footnotesize, inner sep=1.5pt},
  pl-glab/.style ={pl-lab, text=black!45},
  pl-title/.style={font=\small, inner sep=1pt}}
\newcommand{\plframe}[1]{%
  \fill[pl-below] (0.5,0.5) rectangle (10.5,#1);
  \draw[black!15, line width=0.35pt, step=1] (0.5,0.5) grid (10.5,10.5);
  \draw[pl-cut] (0.5,#1) -- (10.5,#1);
  \draw[line width=0.8pt, -{Stealth[length=2mm,width=1.6mm]}]
        (0.25,0.25) -- (11.1,0.25)
        node[font=\footnotesize, anchor=west, inner sep=2pt] {$\lesseq_1$};
  \draw[line width=0.8pt, -{Stealth[length=2mm,width=1.6mm]}]
        (0.25,0.25) -- (0.25,11.1)
        node[font=\footnotesize, anchor=south, inner sep=2pt] {$\lesseq_2$};}

\begin{figure}[H]
\centering
\begin{tikzpicture}[pl]
  \plframe{4.5}
  \node[pl-title] at (5.5,11.7) {$\lambda_i$};
  \node[pl-max]   (l1)  at (10,1) {};
  \node[pl-min]   (lb)  at (8,2)  [label={[pl-lab]right:$\elb$}] {};
  \node[pl-min]   (l3)  at (7,3)  {};
  \node[pl-max]   (lai) at (9,4)  [label={[pl-lab]right:$\ela_i$}] {};
  \node[pl-max]   (lan) at (6,5)  [label={[pl-lab]left:$\ela_{i+1}$}] {};
  \node[pl-min]         at (1,6)  {};
  \node[pl-ghost]       at (4,7)  [label={[pl-glab]left:$\elc$}] {};
  \node[pl-max]         at (5,8)  {};
  \node[pl-pt]    (lw)  at (3,9)  [label={[pl-lab]left:$\elw$}] {};
  \node[pl-min]         at (2,10) {};
  \draw[pl-rel] (lai) to[out=95,in=-5,looseness=0.95] (lw);
  \draw[pl-rel] (lan) to[bend right=12] (l3);

  \begin{scope}[xshift=66mm]
  \plframe{5.5}
  \node[pl-title] at (5.5,11.7) {$\lambda_{i+1}$};
  \node[pl-max]   (r1)  at (10,1) {};
  \node[pl-ghost]       at (8,2)  [label={[pl-glab]right:$\elb$}] {};
  \node[pl-min]   (r3)  at (7,3)  {};
  \node[pl-max]   (rai) at (9,4)  [label={[pl-lab]right:$\ela_i$}] {};
  \node[pl-min]   (ran) at (6,5)  [label={[pl-lab]left:$\ela_{i+1}$}] {};
  \node[pl-min]         at (1,6)  {};
  \node[pl-max]         at (4,7)  [label={[pl-lab]left:$\elc$}] {};
  \node[pl-max]         at (5,8)  {};
  \node[pl-pt]    (rw)  at (3,9)  [label={[pl-lab]left:$\elw$}] {};
  \node[pl-min]         at (2,10) {};
  \draw[pl-rel] (rai) to[out=95,in=-5,looseness=0.95] (rw);
  \draw[pl-rel] (ran) to[bend right=12] (r3);
  \end{scope}
\end{tikzpicture}
\caption{Two consecutive lines for one demand~$\pR$ (red), $m = 1$, one $1$-type; $\pS_2$-edges omitted. Blue/orange: $\domB^{\delta,\min}$/$\domB^{\delta,\max}$; dashed: not in the line; shaded: below it. Crossing the line, $\ela_{i+1}$ evicts~$\elb$ from~$\domB^{\downarrow,\min}$ and makes~$\elc$ $\max$-extremal above; the non-extremal witness~$\elw$ of~$\ela_i$ is~kept.}
\label{fig:processing-lines}
\end{figure}
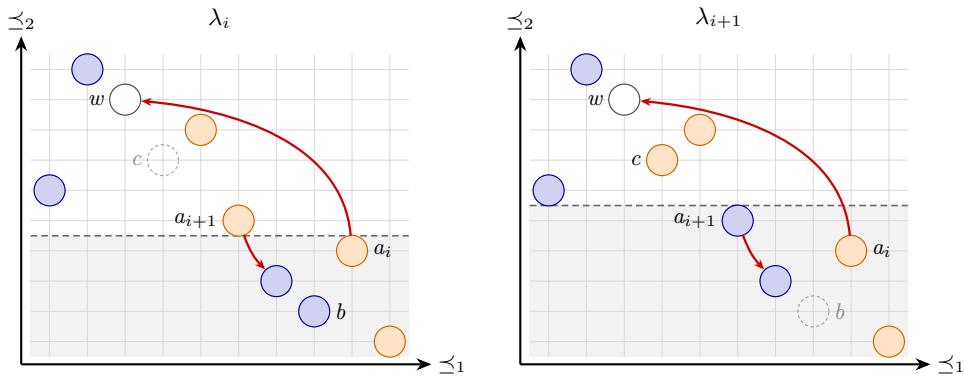

\vspace{-0.75em}

\vspace{0.25em}
\noindent \textbf{Processing lines.}\, Picture a horizontal line just above an element~$\ela$ of a model: the elements below it, $\ela$~included, count as \emph{processed}, and those above it as \emph{to be processed}.
A processing line records what the sweep has to remember at this point.
We define it in terms of~$\varphi$ alone, and extract one from every element of a model afterwards.
Given a finite structure~$\lambda$ with domain~$\domP$ and an element $\ela \in \domP$, we split~$\domP$ into $\domP^{\downarrow} \deff \set{\elb \in \domP \mid \lambda \models \elb \lesseq_2 \ela}$ and $\domP^{\uparrow} \deff \set{\elb \in \domP \mid \lambda \models \ela \less_2 \elb}$.
For $\rmX \subseteq \domP$, the sets $\min(\rmX)$ and~$\max(\rmX)$ consist of the $m{+}1$ leftmost and the $m{+}1$ rightmost elements of~$\rmX$, respectively, or all of~$\rmX$ if it is smaller; and $\rmX(\alpha) \deff \set{\elb \in \rmX \mid \tp{\lambda}{\elb} = \alpha}$ for $\alpha \in \tpone{\varphi}$.

\begin{definition}\label{def:processing-line}
  A \emph{processing line} is a finite structure~$\lambda$ over the signature of~$\varphi$, with domain~$\domP$, equipped with a \emph{distinguished element}~$\ela \in \domP$, with sets $\domB^{\delta,\eta} \subseteq \domP$ for $\delta \in \set{\downarrow, \uparrow}$ and $\eta \in \set{\min, \max}$, whose union we call $\domB$, and, for every $\elb \in \domP$, with a \emph{partial witness set} $\wst{\lambda}{\elb} \subseteq \domP \setminus \set{\elb}$ of at most~$m$ \emph{recorded witnesses} of~$\elb$, such that:
  \begin{namedlist}[(L1)]
    \item[\desclabel{\textup{(L1)}}{cond:validity}] \textbf{Local validity.}
      Both $\lesseq_1$ and~$\lesseq_2$ are linear orders, every realised $1$-type entails~$\chi(\varx, \varx)$, and every realised $2$-type entails~$\chi(\varx, \vary)$.
      Every $\pS_2$-edge points to the immediate successor in the local $\lesseq_2$-order, but elements need not have an $\pS_2$-successor or $\pS_2$-predecessor.
    \item[\desclabel{\textup{(L2)}}{cond:extremal}] \textbf{Extremal elements.}
      $\domB^{\delta,\eta}(\alpha) = \eta(\domP^{\delta}(\alpha))$ for all~$\alpha$, $\delta$ and~$\eta$ (read as an operator).
    \item[\desclabel{\textup{(L3)}}{cond:witnesses}] \textbf{Witnesses at the distinguished element.}
      At the distinguished element all demands are served: for every $t \leqslant m$, some $\elw \in \wst{\lambda}{\ela}$ satisfies $\lambda \models \chi_t(\ela, \elw)$.
  \end{namedlist}
\end{definition}

\noindent Intuitively, $\domP^{\downarrow}$ holds processed elements, whilst $\domP^{\uparrow}$ holds elements to be processed.
The \emph{extremal} sets $\domB^{\delta,\eta}$ hold the $m{+}1$ leftmost and rightmost elements of every $1$-type on either side of the line, much like the royal court of Section~\ref{sec:one-order-one-succ}. In the sequel we write $\domB^{\delta} \deff \domB^{\delta,\min} \cup \domB^{\delta,\max}$.

\vspace{0.25em}
\noindent \textbf{Lines of a model.}\, Let $\strA \models \varphi$.
A \emph{witnessing set} of $\ela \in \domA$ is a set $\rmW \subseteq \domA \setminus \set{\ela}$ such that, for every $t \leqslant m$, some $\elw \in \rmW$ satisfies $\strA \models \chi_t(\ela, \elw)$.
Without loss of generality, we assume that each $\elw \in \rmW$ satisfies some demand.
For every~$\ela$, we fix a witnessing set~$\wst{\strA}{\ela}$ of size at most~$m$, and call its elements \emph{the witnesses} of~$\ela$.
Now fix $\ela \in \domA$.
An element~$\elb$ is \emph{below} the line of~$\ela$ if $\elb \lesseq_2 \ela$, and \emph{above} it otherwise. Note that $\ela$ itself lies below its line, whilst its $\pS_2$-successor, if any, lies above it.
We collect the elements below and above the line in $\domA_{\ela}^{\downarrow}$ and~$\domA_{\ela}^{\uparrow}$.
For $\delta \in \{ \downarrow, \uparrow \}$ write $\domA_{\ela}^{\delta}(\alpha)$ for the elements of $1$-type~$\alpha$ in $\domA_{\ela}^{\delta}$, with types taken in~$\strA$.
For all~$\alpha$, $\delta$ and~$\eta \in \{\min, \max\}$, let $\domB_{\ela}^{\delta,\eta}(\alpha) \deff \eta(\domA_{\ela}^{\delta}(\alpha))$, with extrema taken in~$\strA$; let $\domB_{\ela}^{\delta,\eta}$ be their union over~$\alpha$, and $\domB_{\ela}$ the union over~$\delta$ and~$\eta$.
Write $\opp{\delta}$ for the direction opposite to~$\delta$.
For both~$\delta$, define the \emph{exceptional} sets
$\domC_{\ela}^{\delta,\times} \deff
\set{\elb \in \domA_{\ela}^{\opp{\delta}} \mid
\wst{\strA}{\elb} \cap (\domA_{\ela}^{\delta} \setminus \domB_{\ela}) \neq \emptyset}$
and
$\domC_{\ela}^{\delta,w} \deff
\bigl(\bigcup_{\elb \in \domC_{\ela}^{\delta,\times}} \wst{\strA}{\elb}\bigr)
\cap (\domA_{\ela}^{\delta} \setminus \domB_{\ela})$.
Let $\domC_{\ela}$ be the union of the four sets $\domC_{\ela}^{\delta,\times}$ and~$\domC_{\ela}^{\delta,w}$.
Informally, the elements in~$\domC_{\ela}^{\delta,\times}$ are those whose chosen witnesses across the line cannot all be found among the extrema in~$\domB_{\ela}$.
The role of~$\domC_{\ela}^{\delta,w}$ is to keep these exceptional witnesses in sight.
So $\domB_{\ela}^{\downarrow,\min}(\alpha)$ holds the $m{+}1$ leftmost elements of $1$-type~$\alpha$ below the line, and $\domC_{\ela}^{\downarrow,\times}$ holds the elements above the line that have a witness $\elw \notin \domB_{\ela}$ below it.
Let $\domP_{\ela}$ be the union of~$\set{\ela}$, its $\pS_2$-neighbours, $\wst{\strA}{\ela}$, $\domB_{\ela}$ and~$\domC_{\ela}$.
The \emph{processing line} $\pro{\strA}{\ela}$ of~$\ela$ is the substructure of~$\strA$ induced by $\domP_{\ela}$, equipped with the distinguished element~$\ela$, the sets $\domB_{\ela}^{\delta,\eta}$, and the partial witness sets $\wst{\pro{\strA}{\ela}}{\elb} \deff \wst{\strA}{\elb} \cap \domP_{\ela}$.
It is easy to verify that every element of a model of~$\varphi$ gives rise to a processing~line
(Appendix~\ref{app:processing-lines}, Lemma~\ref{lem:app-model-line}).
In general, the sets~$\domC_{\ela}$ may be as large as the model itself.
The next lemma makes all lines exponentially small by rewiring: if~$\eld$ has a nonextremal
witness~$\elv$ across the line, an extremal element of the same $1$-type lies on the same side
of~$\eld$ in both orders and can take over the $2$-type of~$(\eld, \elv)$; this fails only for
exponentially many~$\eld$.

\begin{lemma}\label{lem:small-lines}
  Let $\strA \models \varphi$.
  Then there is a structure $\strA^{*} \models \varphi$ over the same domain, with the same orders and $1$-types, together with witnessing sets for it, such that, for each $\ela \in \domA$, the domain of~$\pro{\strA^{*}}{\ela}$ has size at most~$2^{|\varphi|^{O(1)}}$.
\end{lemma}
\begin{proof}
We shrink the model-defined sets~$\domC_{\ela}^{\delta,\times}$ one line after another by rewiring witnesses into extremal sets. As will become apparent, the sets for other elements do not grow.

\vspace{0.25em}
\noindent \textbf{Shrinking one line.}\, Fix $\ela \in \domA$ and $\alpha \in \tpone{\varphi}$, and let $\elb_1, \ldots, \elb_k$ be the elements of~$\domC_{\ela}^{\downarrow,\times}$ of $1$-type~$\alpha$.
Each of them lies above the line and has a witness below it outside~$\domB_{\ela}$.
We construct $\strA' \models \varphi$ over the same domain in which $k \leqslant \rmM \deff 2m(m{+}1)|\tpone{\varphi}| + 1$.
The set~$\domB_{\ela}^{\downarrow}$ has at most $2(m{+}1)|\tpone{\varphi}|$ elements, so at most $\rmM{-}1$ of the~$\elb_j$ are witnesses of elements of~$\domB_{\ela}^{\downarrow}$.
If $k > \rmM$, let $\eld_1, \ldots, \eld_r$ be the~$\elb_j$ that are neither such witnesses nor the $\pS_2$-successor of~$\ela$; we spare the latter since we are about to redefine $2$-types, which must not create or destroy $\pS_2$-edges.
Consider such a~$\eld$ and a witness~$\elv$ of~$\eld$ below the line.
If $\elv \in \domB_{\ela}$, we leave it alone.
Otherwise, say $\elv \less_1 \eld$. (The case $\eld \less_1 \elv$ is symmetric, with rightmost instead of leftmost elements.)
As $\elv \notin \domB_{\ela}$, the set $\domB_{\ela}^{\downarrow,\min}(\tp{\strA}{\elv})$ consists of $m{+}1$ elements to the left of~$\elv$, at most $m$ of which are witnesses of~$\eld$.
So we can pick one that is not, say, $\ele$, set $\tp{\strA'}{\eld, \ele} \deff \tp{\strA}{\eld, \elv}$, and fix $\ele$ instead of~$\elv$ as a witness of~$\eld$.
Notice that this operation is legal: $\ele$ and~$\elv$ have the same $1$-type, both lie to the left of and below~$\eld$, and neither is close to~$\eld$, as $\eld$ lies above the $\pS_2$-successor of~$\ela$.
So the orders survive, the imposed $2$-type is realised in~$\strA$ and hence entails~$\chi$, and no witness is lost: $\ele$ was no witness of~$\eld$, and $\eld$ is no witness of~$\ele \in \domB_{\ela}^{\downarrow}$ by the choice of the~$\eld_j$.
Doing so for all witnesses below the line of all~$\eld_j$, and then for all~$\alpha$, we obtain $\strA' \models \varphi$ in which $|\domC_{\ela}^{\downarrow,\times}| \leqslant |\tpone{\varphi}| \cdot \rmM$.
The analogous construction for~$\domC_{\ela}^{\uparrow,\times}$, applied to~$\strA'$ and discounting~$\ela$ from $\elb_1, \ldots, \elb_k$ instead of its successor, yields $\strA'' \models \varphi$ in which both sets $\domC_{\ela}^{\delta,\times}$ are bounded as given above.
As $|\domC_{\ela}^{\delta,w}| \leqslant m|\domC_{\ela}^{\delta,\times}|$ and $|\tpone{\varphi}|$ is exponential in~$|\varphi|$, the line of~$\ela$ in~$\strA''$ has at most~$2^{|\varphi|^{O(1)}}$~elements.

\vspace{0.25em}
\noindent \textbf{Other lines do not grow.}\, The sets~$\domB$ depend only on $1$-types and orders, which have not changed.
Only the elements~$\eld$ removed from $\domC_{\ela}^{\downarrow,\times} \cup \domC_{\ela}^{\uparrow,\times}$ above change their witnesses, so it suffices to show that no such~$\eld$ enters $\domC_{\ela'}^{\downarrow,\times}$ or~$\domC_{\ela'}^{\uparrow,\times}$ for any $\ela' \neq \ela$.
Let $\ela' \less_2 \ela$. (The case $\ela \less_2 \ela'$ is symmetric.)
The elements below the line of~$\ela'$ are among those below the line of~$\ela$, and the elements above the line of~$\ela$ are among those above the line of~$\ela'$.
Hence an element below the line of~$\ela'$ that lies in~$\domB_{\ela}$ lies in~$\domB_{\ela'}$, and an element above the line of~$\ela$ outside~$\domB_{\ela}$ lies outside~$\domB_{\ela'}$.
If $\eld$ is above the line of~$\ela$, then its witnesses below the line of~$\ela'$ now lie in~$\domB_{\ela}$, hence in~$\domB_{\ela'}$, and so $\eld \notin \domC_{\ela'}^{\downarrow,\times}$.
If $\eld$ is below the line of~$\ela$ but above that of~$\ela'$, then only its witnesses above the line of~$\ela$ changed, and these do not matter for~$\domC_{\ela'}^{\downarrow,\times}$.
Finally, if $\eld$ is below the line of~$\ela'$ and now lies in~$\domC_{\ela'}^{\uparrow,\times}$, then it did so before: in~$\strA$, it had a witness above the line of~$\ela$ outside~$\domB_{\ela}$, hence outside~$\domB_{\ela'}$.

\vspace{0.25em}
\noindent \textbf{All lines at once.}\, We treat the elements of~$\domA$ one by one, each time in the current model.
Since no step enlarges the sets $\domC_{\ela}^{\delta,\times}$ for the elements treated before, and $|\domC_{\ela}^{\delta,w}| \leqslant m|\domC_{\ela}^{\delta,\times}|$, the resulting $\strA^{*} \models \varphi$ has all its lines of size at most~$2^{|\varphi|^{O(1)}}$.
\end{proof}

\vspace{0.25em}
\noindent \textbf{Gluing lines.}\, From now on, processing lines have at most~$\rmL \in 2^{O(|\varphi|)}$ elements, the bound of
Lemma~\ref{lem:small-lines} (made explicit in Appendix~\ref{app:processing-lines}), named by
numbers in~$\set{1, \ldots, 2\rmL}$.
A sweep visits the lines of a model bottom-up, so consecutive lines have to fit together.
For a sequence $\lambda_1, \ldots, \lambda_k$ of processing lines, we attach the index~$i$ to the data of~$\lambda_i$, writing $\domP_i$, $\ela_i$, $\domB_i^{\delta,\eta}(\alpha)$, and so on, and take $\min$ and~$\max$ in the line indexed by the $\domB$-set in their argument.
A line is \emph{initial} if $\domP^{\downarrow} = \set{\ela}$, and \emph{final} if $\domP^{\uparrow} = \emptyset$.

\begin{definition}\label{def:gluing-rules}
  A sequence $\lambda_1, \ldots, \lambda_k$ of processing lines satisfies the \emph{gluing rules} if $\lambda_1$ is initial, $\lambda_k$ is final, and the following hold for all~$i < k$:
  \begin{namedlist}[(G1)]
    \item[\desclabel{\textup{(G1)}}{cond:succession}] \textbf{Line succession.}
      $\ela_i, \ela_{i+1} \in \domP_i \cap \domP_{i+1}$, and $\pS_2(\ela_i, \ela_{i+1})$ holds in both $\lambda_i$ and~$\lambda_{i+1}$.
    \item[\desclabel{\textup{(G2)}}{cond:agreement}] \textbf{Agreement and witness continuity.}
      $\lambda_i$ and~$\lambda_{i+1}$ induce the same substructure on $\domP_i \cap \domP_{i+1}$, and $\wst{\lambda_i}{\elb} \cap \domP_{i+1} = \wst{\lambda_{i+1}}{\elb} \cap \domP_i$ for all $\elb \in \domP_i \cap \domP_{i+1}$.
    \item[\desclabel{\textup{(G3)}}{cond:persistence}] \textbf{Persistence of elements.}
      $\domP_i^{\uparrow} \setminus \set{\ela_{i+1}} \subseteq \domP_{i+1}^{\uparrow}$ and $\domP_{i+1}^{\downarrow} \setminus \set{\ela_{i+1}} \subseteq \domP_i^{\downarrow}$.
    \item[\desclabel{\textup{(G4)}}{cond:update}] \textbf{Extremal element update.}
      For $\alpha \in \tpone{\varphi}$, let $\rmT_{i,\alpha} \deff \set{\ela_{i+1}}$ if $\tp{\lambda_i}{\ela_{i+1}} = \alpha$, and $\rmT_{i,\alpha} \deff \emptyset$ otherwise.
      Then, for both~$\eta \in \set{\min, \max}$, when read as an operator,
      \[
        \domB_{i+1}^{\downarrow,\eta}(\alpha) = \eta\bigl(\domB_i^{\downarrow,\eta}(\alpha) \cup \rmT_{i,\alpha}\bigr)
        \quad \text{and} \quad
        \domB_i^{\uparrow,\eta}(\alpha) = \eta\bigl(\domB_{i+1}^{\uparrow,\eta}(\alpha) \cup \rmT_{i,\alpha}\bigr).
      \]
  \end{namedlist}
\end{definition}

\noindent Intuitively, \ref{cond:succession} moves the line up past~$\ela_{i+1}$, \ref{cond:agreement} and~\ref{cond:persistence} make consecutive lines describe the same model, with $\ela_{i+1}$ the only element crossing the line, and~\ref{cond:update} updates the extremal sets as $\ela_{i+1}$~crosses.
As expected, the lines of a model satisfy the gluing~rules, so \emph{if $\varphi$ has a finite
model, then some sequence of processing lines satisfies them} (Appendix~\ref{app:processing-lines},
Lemma~\ref{lem:sat-implies-lines}).
Figure~\ref{fig:processing-lines} shows possible changes across consecutive processing lines.
The converse is the heart of the matter: we glue the lines back into a~model.

\begin{lemma}\label{lem:lines-imply-sat}
  If processing lines $\lambda_1, \ldots, \lambda_k$ satisfy the gluing rules, then $\varphi$ has a finite~model.
\end{lemma}

\begin{proof}
We build a model~$\strA$ whose elements are the distinguished elements of the~lines.

\vspace{0.25em}
\noindent \textbf{Renaming.}\, An \emph{interval of occurrence} of an element is a maximal interval $\lambda_i, \ldots, \lambda_j$ of lines containing it.
We rename every element apart in each of its intervals of occurrence.
That is to say, for each element $b$ and an interval of occurrence $\lambda_i, \ldots, \lambda_j$,
we replace $b$ with a fresh symbol, say, $b_{i, j}$ in $\lambda_i, \ldots, \lambda_j$.
This preserves the gluing rules, so we may speak of \emph{the} interval of occurrence of an~element.
From now on we refer to the distinguished element of $\lambda_i$ as $\ela_i$.

\vspace{0.25em}
\noindent \textbf{Crossings.}\, Take any distinguished element $\ela_i$.
In its interval of occurrence $\lambda_j, \ldots, \lambda_l$
we have that $\ela_i$ is below every $\lambda_i, \ldots, \lambda_l$ whilst being above $\lambda_j, \ldots, \lambda_{i-1}$.
Indeed, $\ela_i$ is below $\lambda_i$ by definition. But then, by induction on $i \leqslant r < l$,
$\lambda_r \models \ela_i \lesseq_2 \ela_r$ by induction, $\lambda_{r} \models \pS_2(\ela_r, \ela_{r+1})$ by \ref{cond:succession},
and $\lambda_r \models \ela_i \less_2 \ela_{r+1}$ by \ref{cond:validity}.
Then, by~\ref{cond:agreement}, $\lambda_{r+1} \models \ela_i \less_2 \ela_{r+1}$ as required.
The \enquote{above} direction is similar. In fact, by the above induction, elements $\ela_1, \ldots, \ela_k$ are distinct.

\vspace{0.25em}
\noindent \textbf{Every element is processed.}\, Every element occurring in a line is some~$\ela_i$.
Otherwise, let $\elb \notin \set{\ela_1, \ldots, \ela_k}$ have interval $\lambda_i, \ldots, \lambda_j$.
Then $\elb$ is above the lines of all of $\lambda_i, \ldots, \lambda_j$ or below all of them: if it were above the line of~$\lambda_l$ but below that of~$\lambda_{l+1}$, then~\ref{cond:persistence} would give $\elb = \ela_{l+1}$.
A crossing from below to above is impossible too: $\elb \lesseq_2 \ela_l$ in~$\lambda_l$ implies $\elb \less_2 \ela_{l+1}$ there by~\ref{cond:succession} and~\ref{cond:validity}, and hence also in~$\lambda_{l+1}$ by~\ref{cond:agreement}.
If $\elb$ is above all of them, then $j = k$, as by~\ref{cond:persistence} it can only leave the lines as~$\ela_{j+1}$, which contradicts $\domP_k^{\uparrow} = \emptyset$; symmetrically, if it is below all of them, then $i = 1$, which contradicts $\domP_1^{\downarrow} = \set{\ela_1}$.

\vspace{0.25em}
\noindent \textbf{Domain and orders.}\, Let $\domA \deff \set{\ela_1, \ldots, \ela_k}$, set $\ela_i \less_2 \ela_j$ for all~$i < j$, so that $\pS_2$ consists of the pairs $(\ela_i, \ela_{i+1})$, and set $\tp{\strA}{\ela_i} \deff \tp{\lambda_i}{\ela_i}$.
We claim the $\lesseq_2$-order defined in $\strA$ coincides with that in processing lines.
To see this take $\ela_j, \ela_l \in \domP_i$ with $j < l$.
If $j < i < l$, then, as argued above, $\ela_j$ is below the line of $\lambda_i$
whilst $\ela_l$ is above it.
Hence, $\lambda_{i} \models \ela_j \less_2 \ela_i \less_2 \ela_l$.
By the above $\lambda_{i} \models \neg \pS_2(\ela_j, \ela_l)$ as well.
In case $i \leqslant j < l$ we see that $\ela_l$ occurs in $\lambda_{j}$.
Since $\ela_j$ is below the line of $\lambda_j$, whilst $\ela_l$ is above it
 we have that $\lambda_{j} \models \ela_j \less_2 \ela_l$.
Now, by \ref{cond:succession}, $\lambda_{j} \models \pS_2(\ela_j, \ela_{j+1})$.
Thus, by \ref{cond:validity}, $\lambda_{j} \models \pS_2(\ela_j, \ela_{l})$ exactly when $l = j{+}1$.
The case $j < l \leqslant i$ is then symmetric.
By~\ref{cond:agreement} this is consistent throughout all $\lambda_i$ in which $\ela_j$ and $\ela_l$ appear.

\noindent
For~$\lesseq_1$, we pick an injection $\bx \colon \domA \to \Q$ and let $\elb \lesseq_1 \elc$ if and only if $\bx(\elb) \leqslant \bx(\elc)$.
We define~$\bx$ on~$\domP_1, \domP_2, \ldots$ in turn.
When we reach~$\domP_{i+1}$, the elements placed so far are those it shares with~$\domP_i$, placed in the order of~$\lambda_i$, and thus also that of~$\lambda_{i+1}$ by~\ref{cond:agreement}.
As $\Q$ is dense, we can place the remaining elements of $\domP_{i+1}$ at fresh positions so that $\bx$ respects the order of~$\lambda_{i+1}$.
Hence $\bx$ respects~$\lesseq_1$ in every~line.

\vspace{0.25em}
\noindent \textbf{Extremal elements.}\, For all $i$, $\alpha$ and~$\eta$, the set~$\domB_i^{\downarrow,\eta}(\alpha)$ consists of the up to $m{+}1$ leftmost ($\eta = \min$) or rightmost ($\eta = \max$) elements of $1$-type~$\alpha$ among $\ela_1, \ldots, \ela_i$ in~$\strA$, and $\domB_i^{\uparrow,\eta}(\alpha)$ of those among $\ela_{i+1}, \ldots, \ela_k$.
This holds for $\domB_1^{\downarrow,\eta}(\alpha)$ and~$\domB_k^{\uparrow,\eta}(\alpha)$ by~\ref{cond:extremal}, as $\domP_1^{\downarrow} = \set{\ela_1}$ and $\domP_k^{\uparrow} = \emptyset$, and propagates by induction, upwards for~$\downarrow$ and downwards for~$\uparrow$, using~\ref{cond:update}.
For instance, $\domB_{i+1}^{\downarrow,\min}(\alpha)$ consists of the $m{+}1$ leftmost elements of $\domB_i^{\downarrow,\min}(\alpha) \cup \rmT_{i,\alpha} \subseteq \domP_i$ in~$\lambda_i$, \ie in~$\bx$, and hence of the $m{+}1$ leftmost elements of $1$-type~$\alpha$ among $\ela_1, \ldots, \ela_{i+1}$.

\vspace{0.25em}
\noindent \textbf{Types of pairs.}\, Distinct elements occurring together in a line take their $2$-type from it.
That is to say, if $\ela, \elb \in \domP_i$, then $\tp{\strA}{\ela, \elb} = \tp{\lambda_i}{\ela, \elb}$.
By~\ref{cond:agreement}, this is well defined, as all lines in the intersection of the elements' intervals of occurrence agree on it.
In particular, by~\ref{cond:witnesses}, each~$\ela_i$ finds in~$\lambda_i$ a witness $\elw \in \wst{\lambda_i}{\ela_i}$ for every~$\chi_t$, and by~\ref{cond:validity}, all types assigned so far entail~$\chi$.
It remains to assign $2$-types to pairs $\ela_i, \ela_j$ with $i < j$ that never occur together.
By~\ref{cond:succession}, $\ela_i$ occurs in~$\lambda_{i+1}$, and hence $\ela_j$ does not; thus $j > i{+}1$, and $\ela_{i+1} \less_2 \ela_j$.
Let $\alpha \deff \tp{\strA}{\ela_j}$.
By the claim on extremal elements, $\domB_{i+1}^{\uparrow,\min}(\alpha)$ and $\domB_{i+1}^{\uparrow,\max}(\alpha)$ consist of $m{+}1$ elements each, all to the left and all to the right of~$\ela_j$, respectively.
Pick $\elb$ from the former if $\bx(\ela_j) < \bx(\ela_i)$, and from the latter otherwise.
Then $\elb$ has the $1$-type of~$\ela_j$ and lies on the same $\lesseq_1$-side of~$\ela_i$, and $\ela_i \less_2 \ela_{i+1} \less_2 \elb$ holds in~$\lambda_{i+1}$, so $\elb$ is not the $\pS_2$-successor of~$\ela_i$.
We may therefore set $\tp{\strA}{\ela_i, \ela_j} \deff \tp{\lambda_{i+1}}{\ela_i, \elb}$, which respects the orders of~$\strA$ and, by~\ref{cond:validity}, entails~$\chi$.
Completing the above we have~$\strA \models \varphi$.
\end{proof}

\begin{proof}[Proof of Theorem~\ref{thm:expspace-easyness}]
Given a sentence, we guess a run of the normal-form procedure in Fact~\ref{fact:scott}, and continue with its output~$\varphi$.
We then guess an initial processing line~$\lambda_1$ and check~\ref{cond:validity}--\ref{cond:witnesses}.
Given~$\lambda_i$, we guess~$\lambda_{i+1}$, check~\ref{cond:validity}--\ref{cond:witnesses} for it and~\ref{cond:succession}--\ref{cond:update} for the pair, and forget~$\lambda_i$.
We accept once we reach a final line.
By Lemmas~\ref{lem:sat-implies-lines} and~\ref{lem:lines-imply-sat}, the algorithm has an accepting run if and only if $\varphi$ has a finite model.
As every line has at most~$\rmL$ elements named in~$\set{1,\ldots,2\rmL}$, this nondeterministic algorithm works in exponential space. By Savitch's theorem, this yields membership in~$\ExpSpace$.
\end{proof}

\section{Conclusions}\label{sec:conclusions}

\noindent We have settled the complexity of the finite satisfiability
problem for~$\FOt$ over two-dimensional structures, as visualised in Table~\ref{tab:results}.
Two linear orders, with or without the successor of one of them, yield an
$\ExpSpace$-complete logic (Theorems~\ref{thm:expspace-hardness}
and~\ref{thm:expspace-easyness}), while one linear order and one (possibly foreign)
successor reduces the complexity to just~$\NExpTime$ (Theorem~\ref{thm:one-order-one-succ}).
As an application, order-invariance and order-successor-invariance of~$\FOt$
are $\coNExpTime$-complete (Theorem~\ref{thm:orderinv}).
All our proofs are relatively simple, which we see as their main~strength.
Our lower bound relies on unary predicates that may hold simultaneously;
whether it survives over data words with an explicit
alphabet~\cite[\S5]{SchwentickZ12} remains~open.
Future work could go in two directions.
The first is to lift our results from finite to arbitrary structures, \eg to
data $\omega$-words.
There, satisfiability of~$\FOt$ with two linear orders and arbitrary auxiliary
symbols is known to be decidable~\cite{TorunczykZ22}, so the task is to pin
down its complexity.
Our lower bound transfers immediately, but the upper bounds need extra work:
instead of a small-model property, one has to establish a small
\emph{ultimately periodic} model~property.
The second leads to a challenging open problem: is (finite)
satisfiability decidable for~$\FOt$ with three or more successor
relations~\cite[\S7,~Question~1]{ZeumeH16}?
We conjecture that it~is.

\vspace{0.25em}
\noindent \textbf{Declaration on generative AI.}\, We developed the main results and first
versions of most proofs without AI assistance. Claude Opus~5.5 helped extensively with
writing, editing and reviewing, and drew our figures in TikZ. In reviewing, it
found an error in an earlier version of Lemma~\ref{lem:shrink} and smaller gaps elsewhere in
Section~\ref{sec:one-order-one-succ}, the gap in Kiero\'{n}ski's proof
(Appendix~\ref{app:runs}), and a missing step in the proof of Lemma~\ref{lem:lines-imply-sat},
and suggested the fixes we adopted. From our drafts, it wrote Appendix~\ref{app:runs},
including Lemmas~\ref{lem:pullback}--\ref{lem:tree}, which claims no novelty and only makes the
paper self-contained. We take full responsibility for the~content.

\clearpage
\appendix
\makeatletter
\edef\thetheorem{\expandafter\noexpand\thesection\@thmcountersep\@thmcounter{theorem}}
\makeatother
\section{Encoding Runs: Proof of Lemma~\ref{lem:runs}}\label{app:runs}

\noindent The construction of~$\varphi_{\autM,u}$ and its correctness over configuration trees
(Lemma~\ref{lem:runs-tree}) follow Kiero\'{n}ski's proof of~\cite[Thm.~5, pp.~1671--1676]{Kieronski06},
with his transitive relation~$\prec$ replaced by~$\desc$; we claim no novelty here.
We reproduce it only for the reader's convenience, and because the original contains a minor
gap, easily repaired, and several typos, which we list first.
The passage from configuration trees to arbitrary models of~$\phitree$ in which~$\desc$ is merely
transitive (Lemmas~\ref{lem:pullback} and~\ref{lem:tree}) is straightforward.

\vspace{0.25em}
\noindent \textbf{Errors in~\cite{Kieronski06}.}\, Kiero\'{n}ski's transition formul{\ae} do not
cover a head standing on cell~$0$.
His formul{\ae}~$\Phi^{\mathit{exist}}_{a,q}$
and~$\Phi^{\mathit{univ}}_{a,q}$~\cite[pp.~1674--1675]{Kieronski06}, which
describe transitions, quantify over the cell~$y$ of a child configuration one
to the left of the head, \ie with~$\bar{P}(x) = \bar{P}(y){+}1$.
If the head is on cell~$0$, no such~$y$ exists, and both formul{\ae} hold
vacuously.
This happens already in the initial configuration, whose head is on cell~$0$
by his~$\Phi_7$.
His assumption that the head never moves left of its initial position does
not help, as the formul{\ae} anchor~$y$ one cell to the left of the head for
right moves as~well.
Hence the children of the root configuration may carry no head at all; then
no transition formula ever applies again, his~$\Phi_{10}$ merely copies the
tape, and his~$\Phi_{11}$, which forbids the rejecting state, holds trivially.
Consequently, his~$\Phi$ (with the typo in~$\Phi_4$ below corrected) is
satisfiable whenever the initial state is not rejecting, regardless of whether
the machine accepts, so the reduction fails as stated.
The repair is easy: we locate~$y$ at the cell under the head, and require
explicitly that the cell receiving the head exists (see~$\psi_t$ below).
This also treats all combinations of head moves uniformly,
whereas~\cite[p.~1674]{Kieronski06} spells out only one of them.
The remaining issues are~minor.
\begin{itemize}\itemsep0em
  \item $\Phi_4$~\cite[p.~1673]{Kieronski06} reads
    $\bigwedge_i \forall x\, (A_i x \to \bigwedge_{i \neq j} \neg A_i x)$,
    which makes every~$A_i$ empty and thus his~$\Phi$ unsatisfiable;
    $\neg A_j x$ is~meant.
  \item The left-move parts of $\Phi^{\mathit{exist}}_{a,q}$
    and~$\Phi^{\mathit{univ}}_{a,q}$ end with
    $\forall x\, (y \prec x \to \mathrm{Next}(y,x) \to A''y)$, which
    writes~$a''$ into the cell receiving the head; $A''x$ is~meant.
  \item The machine accepts or rejects exactly in the $2^{n^k}$-th
    step~\cite[p.~1671]{Kieronski06}, but the deepest configurations of the
    tree have number~$2^{n^k}{-}1$ (his~$\Phi_3$), so runs are cut one step
    before they halt, and $\Phi_{11}$ never sees a halting configuration;
    halting at step~$2^{n^k}{-}1$ is~meant.
  \item $\Phi_6$~\cite[p.~1674]{Kieronski06} allows the head cell to carry
    several states. This is harmless, but we forbid it, as our proof reads
    off one state per~configuration.
  \item In the correctness proof~\cite[p.~1676]{Kieronski06}, the sequence
    $b_0, \ldots, b_{2^{n^k}}$ should end with~$b_{2^{n^k}-1}$, the
    transition formul{\ae} are called~$\Phi^{\mathit{exists}}_{a,q}$ instead
    of~$\Phi^{\mathit{exist}}_{a,q}$, and \enquote{an accepting computation
    tree of~$\mathcal{M}$ on~$w$} names the model~$\mathcal{M}$ where the
    machine~$M$ is~meant.
\end{itemize}
Apart from these repairs, we deviate from~\cite{Kieronski06} only in
presentation.
The run encoding works over the given tree, so it needs neither his axioms
for its shape ($\Phi_1$--$\Phi_3$), which we adapt in~$\phitree$ instead, nor
his auxiliary predicate~$N$.
His $\mathrm{Next}$ comes with a separate guard~$x \prec y$, which we fold
into it, and our macro~$\fDown$ is new.
Finally, as the tree is given, we need not extract configurations from
witnesses, and we derive acceptance by a bottom-up induction instead of
completing the tree with dummy~subtrees.

\vspace{0.25em}
\noindent \textbf{The machine.}\, Let $\autM$ work in time~$2^{n^k}$ on
inputs of length~$n$, and let $u = u_0 \cdots u_{n-1}$ be an input.
The states of~$\autM$ are existential, universal, accepting, or rejecting; the
latter two are \emph{halting}.
A configuration is \emph{accepting} if its state is accepting, or existential
and some successor configuration is accepting, or universal and all successor
configurations are accepting; $\autM$ \emph{accepts}~$u$ if its initial
configuration on~$u$ is accepting.
Following~\cite[p.~1671]{Kieronski06}, we assume without loss of generality
that $q_0$ is the initial and~$q_{\mathrm{rej}}$ the only rejecting state,
that every non-halting configuration has exactly two successor
configurations, and that the head never moves left of its initial position;
this affects the running time only by a constant factor.
Hence there is a constant~$c$ such that, for $\rmN \deff n^k{+}c$, every
computation path of~$\autM$ on~$u$ has fewer than $2^{\rmN}$~steps.
Thus every computation tree of~$\autM$ on~$u$ fits into the
$\rmN$-configuration tree~$\strA$: configurations at depth~$d$ are those
reached after $d$~steps, configurations at depth~$2^{\rmN}{-}1$ are halting,
and the head never leaves the cells $0, \ldots, 2^{\rmN}{-}1$.

\noindent We describe configurations in the standard way, by unary
predicates~$\pA_a$ for letters~$a$, $\pQ_q$ for states~$q$, and~$\pH$, which
form~$\sigaux$: every cell carries the predicate of its letter, and the cell
under the head also carries~$\pH$ and the predicate of the state.
The blank letter is~$\sqcup$.
A transition~$t = (q', a', d)$ writes~$a'$, enters~$q'$, and moves the head
by~$d \in \set{-1, +1}$; we write~$t(c)$ for the configuration obtained
from~$c$ by applying~$t$.
For every non-halting state~$q$ and letter~$a$, we fix an order
$t^1_{q,a}, t^2_{q,a}$ of the two transitions applicable in~$q$
reading~$a$.

\vspace{0.25em}
\noindent \textbf{Navigation.}\, For fixed~$\ell < 2^{\rmN}$, the comparisons
$\pD(\varx) = \ell$, $\pC(\varx) = \ell$, $\pC(\varx) \geqslant \ell$,
$\pD(\varx) = \pD(\vary)$, $\pC(\varx) = \pC(\vary)$,
$\pD(\vary) = \pD(\varx){+}1$ and~$\pC(\vary) = \pC(\varx){+}1$ of the depths
and positions encoded at~$\varx$ and~$\vary$ (see Section~\ref{sec:lower-bound})
are expressible by quantifier-free formul{\ae} of size polynomial in~$\rmN$.
For instance, as bits are numbered most significant first, the last one
is~expressed~by
\[
  \textstyle\bigvee_{i < \rmN} \big( \neg\pC_i(\varx) \land \pC_i(\vary) \land
  \bigwedge_{j > i} (\pC_j(\varx) \land \neg\pC_j(\vary)) \land
  \bigwedge_{j < i} (\pC_j(\varx) \iff \pC_j(\vary)) \big).
\]
With their help, we define two navigation~macros:
\begin{align*}
  \fNext(\varx, \vary) &\deff \varx \desc \vary \land \pD(\varx) = \pD(\vary)
    \land \pC(\vary) = \pC(\varx){+}1, \\
  \fDown(\varx, \vary) &\deff \varx \desc \vary \land \pD(\vary) = \pD(\varx){+}1
    \land \pC(\vary) = \pC(\varx).
\end{align*}
Thus $\fNext(\varx, \vary)$ says that $\vary$ is the next cell after~$\varx$,
and $\fDown(\varx, \vary)$ that $\vary$ has the position of~$\varx$ and lies in
a child of its configuration, the left one if and only if~$\pL(\vary)$.
The conjunct $\varx \desc \vary$ does the real work: without it, $\vary$ could
lie in any configuration of the right~depth.

\vspace{0.25em}
\noindent \textbf{The formula.}\, The formul{\ae} $\Phi_1$--$\Phi_3$ say that
every cell carries exactly one letter, that every configuration has at most
one head, and that exactly the head cells carry a state, one~each:
\begin{align*}
  \Phi_1 &\deff \forall\varx\, \textstyle\bigvee_{a} \big( \pA_a(\varx) \land
    \bigwedge_{\hat{a} \neq a} \neg\pA_{\hat{a}}(\varx) \big), \\
  \Phi_2 &\deff \forall\varx\forall\vary\, \big( (\varx \desc \vary \land
    \pD(\varx) = \pD(\vary)) \to \neg(\pH(\varx) \land \pH(\vary)) \big), \\
  \Phi_3 &\deff \forall\varx\, \big( (\pH(\varx) \iff \textstyle\bigvee_{q}
    \pQ_q(\varx)) \land \bigwedge_{q \neq \hat{q}} \neg(\pQ_q(\varx) \land
    \pQ_{\hat{q}}(\varx)) \big).
\end{align*}
The formul{\ae} $\Phi_4$--$\Phi_6$ say that the root configuration is the
initial configuration of~$\autM$ on~$u$, and $\Phi_7$ that cells not under the
head keep their letters in both child~configurations:
\begin{align*}
  \Phi_4 &\deff \forall\varx\, \big( (\pD(\varx) = 0 \land \pC(\varx) = 0) \to
    (\pH(\varx) \land \pQ_{q_0}(\varx)) \big), \\
  \Phi_5 &\deff \textstyle\bigwedge_{i < n} \forall\varx\, \big( (\pD(\varx) = 0
    \land \pC(\varx) = i) \to \pA_{u_i}(\varx) \big), \\
  \Phi_6 &\deff \forall\varx\, \big( (\pD(\varx) = 0 \land \pC(\varx) \geqslant n)
    \to \pA_{\sqcup}(\varx) \big), \\
  \Phi_7 &\deff \forall\varx\forall\vary\, \big( (\fDown(\varx, \vary) \land
    \neg\pH(\varx)) \to \textstyle\bigwedge_{a} (\pA_a(\varx) \iff \pA_a(\vary))
    \big).
\end{align*}
For a transition~$t = (q', a', d)$, the formula~$\psi_t(\vary)$ says
that~$\vary$ looks like the formerly scanned cell after applying~$t$: it
carries~$a'$, and the cell receiving the head, \ie the next cell after~$\vary$
if~$d = +1$ and the previous one if~$d = -1$, exists and carries the head in
state~$q'$:
\begin{align*}
  \psi_{(q', a', +1)}(\vary) &\deff \pA_{a'}(\vary) \land
    \pC(\vary) \neq 2^{\rmN}{-}1 \land \forall\varx\, \big(
    \fNext(\vary, \varx) \to (\pH(\varx) \land \pQ_{q'}(\varx)) \big), \\
  \psi_{(q', a', -1)}(\vary) &\deff \pA_{a'}(\vary) \land
    \pC(\vary) \neq 0 \land \forall\varx\, \big(
    \fNext(\varx, \vary) \to (\pH(\varx) \land \pQ_{q'}(\varx)) \big).
\end{align*}
Let $q$ be a non-halting state, $a$ a letter, and~$t^i \deff t^i_{q,a}$.
If $q$ is existential, $\Phi^{\exists}_{q,a}$ enforces that the left child
configuration arises by~$t^1$ or by~$t^2$, and says nothing about the right
one.
If $q$ is universal, $\Phi^{\forall}_{q,a}$ enforces that the left child
arises by~$t^1$ and the right one by~$t^2$:
\begin{align*}
  \Phi^{\exists}_{q,a} &\deff \forall\varx\forall\vary\, \big(
    (\pA_a(\varx) \land \pQ_q(\varx) \land \pH(\varx) \land
    \fDown(\varx, \vary) \land \pL(\vary)) \to (\psi_{t^1}(\vary) \lor
    \psi_{t^2}(\vary)) \big), \\
  \Phi^{\forall}_{q,a} &\deff \forall\varx\forall\vary\, \big(
    (\pA_a(\varx) \land \pQ_q(\varx) \land \pH(\varx) \land
    \fDown(\varx, \vary)) \\
  &\qquad\qquad\to ((\pL(\vary) \to \psi_{t^1}(\vary)) \land
    (\neg\pL(\vary) \to \psi_{t^2}(\vary))) \big).
\end{align*}
Finally, $\Phi_8 \deff \forall\varx\, \neg\pQ_{q_{\mathrm{rej}}}(\varx)$
forbids rejection, and we~put
\[
  \varphi_{\autM,u} \deff \textstyle\bigwedge_{i \leqslant 8} \Phi_i \land
  \bigwedge_{q \text{ existential},\, a} \Phi^{\exists}_{q,a} \land
  \bigwedge_{q \text{ universal},\, a} \Phi^{\forall}_{q,a}.
\]
Clearly, $\varphi_{\autM,u}$ is computable in time polynomial in~$|\autM|$
and~$n$.
A sentence is \emph{universal} if its negation normal form has no
existential quantifier, and $\desc$ \emph{occurs only negatively} in it if, in
negation normal form, it occurs only in literals~$\neg(\varx \desc \vary)$.
By construction, $\varphi_{\autM,u}$ is universal, equality-free, and
uses~$\desc$ only negatively, namely in antecedents of implications, directly
or within~$\fNext$ and~$\fDown$.
Its correctness over configuration trees reads as~follows.

\begin{lemma}\label{lem:runs-tree}
  $\autM$ accepts~$u$ if and only if some expansion of~$\strA$ by~$\sigaux$
  satisfies~$\varphi_{\autM,u}$.
\end{lemma}

\begin{proof}
We identify a prefix~$w$ with the configuration
$\set{ww' \colon |w'| = \rmN}$ of~$\strA$, whose children are~$w0$, the left
one, and~$w1$.
In an expansion of~$\strA$ by~$\sigaux$, the configuration~$w$
\emph{encodes} a configuration~$c$ of~$\autM$ if the only predicates
of~$\sigaux$ holding in cells of~$w$ are~$\pA_a$ in the cell of each
position~$i$, where~$a$ is the $i$-th letter of~$c$, and~$\pH$ and~$\pQ_q$ in
the cell of position~$p$, where the head of~$c$ is on cell~$p$ in state~$q$.

($\Rightarrow$) Let $\autM$ accept~$u$.
Top-down, we assign to each configuration~$w$ of~$\strA$ either an accepting
configuration~$c_w$ of~$\autM$, or the label \emph{junk} together with a
sequence of letters.
The root gets the initial configuration, which is accepting.
Let $c_w$ be non-halting, so that $|w| < 2^{\rmN}{-}1$, with state~$q$ and
letter~$a$ under the head.
If $q$ is universal, then $t^1_{q,a}(c_w)$ and~$t^2_{q,a}(c_w)$ are
accepting, and we assign them to $w0$ and~$w1$, respectively.
If $q$ is existential, then $t(c_w)$ is accepting for some
$t \in \set{t^1_{q,a}, t^2_{q,a}}$; we assign it to~$w0$, and declare $w1$
and all configurations below it junk, carrying the letters of~$c_w$.
If $c_w$ is halting, all configurations below~$w$, if any, are junk carrying
the letters of~$c_w$.
We expand~$\strA$ so that every~$w$ encodes its~$c_w$, while junk
configurations carry their letters, but neither head nor~state.
Then $\Phi_1$--$\Phi_6$ hold by construction, and $\Phi_8$ holds as no~$c_w$
is rejecting.
For~$\Phi_7$, note that~$c_w$ and the configurations assigned to the children
of~$w$ differ at most in the cell under the head of~$c_w$, and that junk
copies the letters of its parent.
Finally, heads occur only in configurations~$w$ with $c_w$ defined, and if
$c_w$ is non-halting, its children encode the successors of~$c_w$ prescribed
above.
As the head never leaves the cells $0, \ldots, 2^{\rmN}{-}1$, the cells
receiving the heads exist, so all $\Phi^{\exists}_{q,a}$
and~$\Phi^{\forall}_{q,a}$~hold.

($\Leftarrow$) Let an expansion~$\strA'$ of~$\strA$
satisfy~$\varphi_{\autM,u}$.
By~$\Phi_1$--$\Phi_6$, the root encodes the initial configuration
of~$\autM$ on~$u$.
Let~$w$ encode a non-halting configuration~$c$ reached after $|w|$~steps,
whose head is in state~$q$ on cell~$p$, reading~$a$.
Then $|w| < 2^{\rmN}{-}1$, and for the cells~$\elx$ of~$w$, $\ely_0$ of~$w0$
and~$\ely_1$ of~$w1$ of position~$p$, we have
$\strA' \models \fDown(\elx, \ely_j)$, with $\ely_0 \in \pL$
and~$\ely_1 \notin \pL$.
If $\strA' \models \psi_t(\ely_j)$ for a transition~$t$, then $wj$
encodes~$t(c)$: $\psi_t$ fixes the letter of~$\ely_j$ and, by its guard, the
cell receiving the head with its state; $\Phi_1$--$\Phi_3$ exclude further
letters, heads, and states; and $\Phi_7$ fixes the remaining letters.
Thus if $q$ is universal, $\Phi^{\forall}_{q,a}$ makes $w0$ and~$w1$ encode
$t^1_{q,a}(c)$ and~$t^2_{q,a}(c)$, \ie all successors of~$c$; if $q$ is
existential, $\Phi^{\exists}_{q,a}$ makes $w0$ encode a successor of~$c$.
Let $\rmR$ collect the configurations of~$\strA$ reachable from the root by
passing from those encoding universal configurations to both children, and
from those encoding existential ones to the left child.
By the above, every $w \in \rmR$ encodes a configuration~$c_w$ reached from
the initial one after $|w|$~steps; in particular, $c_w$ is halting
if~$|w| = 2^{\rmN}{-}1$.
By induction on $2^{\rmN}{-}1{-}|w|$, every such~$c_w$ is accepting: if
$c_w$ is halting, it is not rejecting by~$\Phi_8$; otherwise, all successors
of~$c_w$ (if universal), or one of them (if existential), are encoded by
children of~$w$ in~$\rmR$, which are accepting by induction.
For the root, this means that $\autM$ accepts~$u$.
\end{proof}

\vspace{0.25em}
\noindent \textbf{From configuration trees to arbitrary models.}\, Models
of~$\phitree$ need not be configuration trees, as witnesses may coincide, but
it suffices that they contain an image of~$\strA$ that is faithful on unary
predicates, in the following~sense.

\begin{definition}\label{def:tree-hom}
  A \emph{tree homomorphism} into a $\sigenc$-structure~$\strB$ is a map
  $\homoh \colon \domA \to \domB$ such that, for all $\ela, \elb \in \domA$
  and every unary~$\pP \in \sigenc$, we have $\ela \in \pP^{\strA}$ if and
  only if $\homoh(\ela) \in \pP^{\strB}$, and $\ela \desc \elb$
  implies~$\homoh(\ela) \desc \homoh(\elb)$.
\end{definition}

\noindent Universal sentences using~$\desc$ only negatively pull back along
tree~homomorphisms.

\begin{lemma}\label{lem:pullback}
  Let $\homoh$ be a tree homomorphism into~$\strB$, let $\strB'$ be an
  expansion of~$\strB$ by~$\sigaux$, and let $\strA'$ be the expansion
  of~$\strA$ by~$\sigaux$ with $\pP^{\strA'} \deff \homoh^{-1}(\pP^{\strB'})$
  for every~$\pP \in \sigaux$.
  Then $\strA'$ satisfies all universal equality-free sentences that hold
  in~$\strB'$ and use~$\desc$ only~negatively.
\end{lemma}

\begin{proof}
  We show, by induction on a formula~$\theta$ of this kind in negation normal
  form, that $\strB' \models \theta[\homoh \circ \nu]$ implies
  $\strA' \models \theta[\nu]$ for every assignment~$\nu$ into~$\domA$.
  Unary literals are preserved and reflected by~$\homoh$: those
  over~$\sigenc$ since $\homoh$ is a tree homomorphism, and those
  over~$\sigaux$ by the definition of~$\strA'$.
  Literals $\neg(\varx \desc \vary)$ transfer since $\homoh$
  preserves~$\desc$, and conjunctions and disjunctions are immediate.
  Finally, if $\strB' \models \forall\varx\, \theta\,[\homoh \circ \nu]$,
  then $\strB' \models \theta\,[\homoh \circ \nu[\varx \mapsto \ela]]$ for
  every~$\ela \in \domA$, so $\strA' \models \forall\varx\, \theta\,[\nu]$
  by~induction.
\end{proof}

\noindent Equality and positive occurrences of~$\desc$ would break the
argument, since $\homoh$ need not be injective and need not reflect~$\desc$.
Tree homomorphisms, in turn, are provided by~$\phitree$.

\begin{lemma}\label{lem:tree}
  We have $\strA \models \phitree$.
  Moreover, for every $\sigenc$-structure~$\strB \models \phitree$
  interpreting~$\desc$ as a transitive relation, there is a tree homomorphism
  into~$\strB$.
\end{lemma}

\begin{proof}
  The conjuncts of~$\phitree$ hold in~$\strA$ by
  Definition~\ref{def:configuration-tree}, with witnesses~$0^{\rmN}$, the
  next cell, and the first cells $w0\,0^{\rmN}$ and~$w1\,0^{\rmN}$ of the
  children of~$w$ for the last cell~$w1^{\rmN}$, respectively.

  Now let $\strB \models \phitree$ with~$\desc^{\strB}$ transitive.
  We define~$\homoh$ by increasing depth and, within a configuration, by
  increasing position, maintaining that $\homoh(ww')$ agrees with~$ww'$ on
  depth, position, and membership in~$\pL$.
  Let $\homoh(0^{\rmN})$ be a witness of the first conjunct.
  As the image of the next cell after~$ww'$, for~$w' \neq 1^{\rmN}$, we take
  a witness of the second conjunct for~$\homoh(ww')$; as the image
  of~$wb\,0^{\rmN}$, for~$b \in \set{0,1}$, we take a witness of the third
  conjunct for~$\homoh(w1^{\rmN})$ and~$b$.
  The invariant is maintained since $wb\,0^{\rmN} \in \pL^{\strA}$ if and only
  if~$b = 0$.
  As the~$\pD_i$ and~$\pC_i$ encode depths and positions exactly, the
  invariant means that $\homoh$ preserves and reflects all unary predicates.
  Finally, $\desc^{\strA}$ is the transitive closure of the pairs of a cell
  and its next cell, and of the pairs~$(w1^{\rmN}, wb\,0^{\rmN})$.
  As $\homoh$ maps them into the transitive relation~$\desc^{\strB}$, it
  preserves~$\desc$.
\end{proof}

\begin{proof}[Proof of Lemma~\ref{lem:runs}]
  Take $c$, $\sigaux$, and~$\varphi_{\autM,u}$ as above.
  By Lemmas~\ref{lem:runs-tree} and~\ref{lem:tree}, (i)~implies~(ii), which
  implies~(iii) as $\strA$ is finite and $\desc^{\strA}$ is transitive.
  Conversely, let $\strB'$ be a finite model of
  $\phitree \land \varphi_{\autM,u}$ in which~$\desc$ is transitive.
  Lemma~\ref{lem:tree} yields a tree homomorphism into the $\sigenc$-reduct
  of~$\strB'$, along which~$\varphi_{\autM,u}$, which is universal,
  equality-free, and uses~$\desc$ only negatively, pulls back to an expansion
  of~$\strA$ by Lemma~\ref{lem:pullback}; so $\autM$ accepts~$u$ by
  Lemma~\ref{lem:runs-tree}.
\end{proof}

\section{Processing Lines: Model-Induced Lines}\label{app:processing-lines}

Fix a finite $\strA \models \varphi$ and its witnessing sets as in
Section~\ref{sec:upper-bound}.
For $\ela \in \domA$, the set $\domB_{\ela}$ collects the $m{+}1$ leftmost and rightmost elements
of each $1$-type on each side of the line of~$\ela$.
The set $\domC_{\ela}$ consists of elements with a selected witness across
that line outside~$\domB_{\ela}$, together with those witnesses.
Recall that $\domP_{\ela}$ consists of~$\ela$, its $\pS_2$-neighbours,
the witnesses $\wst{\strA}{\ela}$, and the elements of
$\domB_{\ela} \cup \domC_{\ela}$.

\begin{lemma}\label{lem:app-model-line}
  For every $\ela \in \domA$, the structure $\pro{\strA}{\ela}$ is a
  processing line satisfying \ref{cond:validity}--\ref{cond:witnesses}.
\end{lemma}

\begin{proof}
Fix~$\ela$ and put $\lambda \deff \pro{\strA}{\ela}$ and
$\domP \deff \domP_{\ela}$. Note that
$\domP^{\delta} = \domP \cap \domA_{\ela}^{\delta}$.
All types and orders of~$\lambda$ are inherited from~$\strA$, and every
retained $\pS_2$-edge joins consecutive elements in~$\strA$, and hence also
in~$\lambda$. This establishes~\ref{cond:validity}.
For~\ref{cond:extremal}, fix $\alpha$, $\delta$ and~$\eta$.
Since $\domB_{\ela}^{\delta,\eta}(\alpha)
= \eta(\domA_{\ela}^{\delta}(\alpha))$ lies in~$\domP$, restricting to~$\domP$
leaves these $m{+}1$ extrema unchanged. Hence
$\eta(\domP^{\delta}(\alpha))
= \domB_{\ela}^{\delta,\eta}(\alpha)$.
Finally, for~\ref{cond:witnesses},
$\wst{\lambda}{\ela}=\wst{\strA}{\ela}$, and the underlying structure
of~$\lambda$ is $\strA \restr \domP$.
As the $\chi_t$ are quantifier-free, these witnesses still serve all
demands of~$\ela$.
\end{proof}

\begin{lemma}\label{lem:app-model-gluing}
  Let $\ela_1 \less_2 \cdots \less_2 \ela_k$ be an enumeration of~$\domA$ and set
  $\lambda_i \deff \pro{\strA}{\ela_i}$.
  Then~$\lambda_1,\ldots,\lambda_k$ satisfies the gluing rules
  \ref{cond:succession}--\ref{cond:update}.
\end{lemma}

\begin{proof}
The first and last lines are initial and final, as $\ela_1$ is bottommost and~$\ela_k$ topmost.
Fix~$i<k$ and put $\els \deff \ela_{i+1}$.
Write $\domC_j^{\delta,\times}$ and~$\domC_j^{\delta,w}$ for the model-defined
sets $\domC_{\ela_j}^{\delta,\times}$ and~$\domC_{\ela_j}^{\delta,w}$.
Only~$\els$ changes sides between the two lines.

\noindent \textbf{Succession and extrema.}\,
The element~$\els$ is the $\pS_2$-successor of~$\ela_i$ in~$\strA$.
Both lines contain these elements and inherit their $\pS_2$-edge,
giving~\ref{cond:succession}.
For either $\eta \in \set{\min,\max}$ and any finite ordered set~$\rmX$,
$\eta(\rmX \cup \set{\els})
= \eta(\eta(\rmX) \cup \set{\els})$.
Applying this identity to each $1$-type on both sides of the two lines
gives the equations of~\ref{cond:update}.

\noindent \textbf{Persistence.}\,
Deleting~$\els$ preserves every other extremal element, while adding it
cannot make a nonextremal element extremal.
Consider $\elb \in \domP_i^{\uparrow} \setminus \set{\els}$.
If it is extremal, it remains so in $\lambda_{i+1}$. Now, if $\elb$ is a witness of~$\ela_i$ or of an element
of~$\domC_i^{\uparrow,\times}$, it lies in~$\domB_{i+1}$ or
$\domC_{i+1}^{\uparrow,w}$.
Lastly, if $\elb \in \domC_i^{\downarrow,\times}$, its nonextremal witness below the first line remains
nonextremal below the second, so $\elb \in \domC_{i+1}^{\downarrow,\times}$.
These exhaust the construction of~$\domP_i^{\uparrow}$ and prove the first
inclusion in~\ref{cond:persistence}.

Conversely, let $\elb \in \domP_{i+1}^{\downarrow} \setminus \set{\els}$.
An extremal element stays extremal when~$\els$ is removed.
Now, if $\elb$ is a witness of~$\els$ or of an element of~$\domC_{i+1}^{\downarrow,\times}$,
it lies in~$\domB_i$ or~$\domC_i^{\downarrow,w}$.
And again, if $\elb \in \domC_{i+1}^{\uparrow,\times}$, its nonextremal witness above the line of $\lambda_{i+1}$ remains
nonextremal above the line of $\lambda_i$. Thus, $\elb \in \domC_i^{\uparrow,\times}$.
Finally, by definition, $\ela_i \in \domP_i$.
This proves the second inclusion.

\noindent \textbf{Agreement and witnesses.}\,
Both lines are induced substructures of~$\strA$, so they agree on their
common domain.
For $\elb \in \domP_i \cap \domP_{i+1}$, both
$\wst{\lambda_i}{\elb} \cap \domP_{i+1}$ and
$\wst{\lambda_{i+1}}{\elb} \cap \domP_i$ equal
$\wst{\strA}{\elb} \cap \domP_i \cap \domP_{i+1}$.
This proves~\ref{cond:agreement}.
\end{proof}

\begin{lemma}\label{lem:sat-implies-lines}
  If $\varphi$ has a finite model, then some sequence of processing lines, each with at
  most~$\rmL$ elements named in~$\set{1, \ldots, 2\rmL}$, satisfies the gluing~rules.
\end{lemma}

\begin{proof}
Take $\strA^{*}$ from Lemma~\ref{lem:small-lines}, and let $\rmM$ be as in its proof.
A line of~$\strA^{*}$ consists of its distinguished element, at most two $\pS_2$-neighbours, at
most $m$~witnesses, at most $4(m{+}1)|\tpone{\varphi}|$ extremal elements, and at most
$|\tpone{\varphi}|\rmM$ and $m|\tpone{\varphi}|\rmM$ elements in each set
$\domC^{\delta,\times}$ and~$\domC^{\delta,w}$, respectively.
So it has at most
$\rmL \deff 3 + m + 4(m{+}1)|\tpone{\varphi}| + 2(m{+}1)|\tpone{\varphi}|\rmM \in 2^{O(|\varphi|)}$
elements.
Let $\ela_1 \less_2 \cdots \less_2 \ela_k$ be the elements of~$\strA^{*}$ and
$\lambda_i \deff \pro{\strA^{*}}{\ela_i}$; by Lemma~\ref{lem:app-model-gluing}, they satisfy the
gluing rules.
Rename their elements line by line, giving each element of $\domP_{i+1} \setminus \domP_i$ a
number unused in~$\lambda_i$ and~$\lambda_{i+1}$; one is free, as fewer than $2\rmL$ are taken.
The renaming is injective on every $\domP_i \cup \domP_{i+1}$, so it preserves the~rules.
\end{proof}

\end{document}